\documentclass[12pt,letterpaper]{JHEP3}

\usepackage{amscd,amsmath,amssymb,amsfonts,xspace,mathrsfs,amsthm}
\usepackage{color, mathtools}
\usepackage{url}

\usepackage{latexsym}
\usepackage{graphicx}
\usepackage{dsfont}
\usepackage{color}
\usepackage{longtable}

\usepackage{epsf}
\usepackage[latin1]{inputenc}
\usepackage{multirow}
\usepackage{epsfig}

\newtheorem{theorem}{Theorem}
\newtheorem{proposition}{Proposition}

\newtheorem{assume}{Assumption}

\numberwithin{equation}{section}

\def\bea{\begin{eqnarray}}
\def\eea{\end{eqnarray}}
\def\be{\begin{equation}}
\def\ee{\end{equation}}
\def\ba{\begin{align}}
\def\ea{\end{align}}
\def\bse{\begin{subequations}}
\def\ese{\end{subequations}}

\newcommand{\nn}{\nonumber}

\def\det{\,{\rm det}\, }

\DeclareMathOperator{\Td}{Td}
\DeclareMathOperator{\ch}{ch}

\DeclareMathOperator{\Erfc}{Erfc}
\newcommand{\Tr}{\mbox{Tr}}

\newcommand{\sgn}{\mbox{sgn}}

\def\({\left(}
\def\){\right)}
\def\[{\left[}
\def\]{\right]}
\def\<{\left\langle}
\def\>{\right\rangle}
\def\hf{{1\over 2}}

\newcommand{\p}{\partial}

\newcommand{\vth}{\vartheta}

\newcommand{\eps}{\epsilon}

\newcommand{\I}{\mathrm{i}}

\newcommand{\cD}{\mathcal{D}}
\newcommand{\cF}{\mathcal{F}}

\newcommand{\cS}{\mathcal{S}}

\newcommand{\cM}{\mathcal{M}}

\newcommand{\cN}{\mathcal{N}}
\newcommand{\cE}{\mathcal{E}}

\newcommand{\cT}{\mathcal{T}}

\newcommand{\cI}{\mathcal{I}}
\newcommand{\cO}{\mathcal{O}}

\newcommand{\IR}{\mathds{R}}

\newcommand{\IZ}{\mathds{Z}}

\newcommand{\IN}{\mathds{N}}

\newcommand{\IP}{\mathds{P}}

\def\scM{\mathscr{M}}

\def\ba{\bar a}

\def\bw{\bar w}

\def\bM{\bar M}

\def\btau{\bar \tau}

\def\bOm{\bar\Omega}

\def\hq{\hat q}

\def\CY{\mathfrak{Y}}

\def\hint{h^{\rm (int)}}
\def\hpol{h^{\rm (p)}}

\def\han{h^{\rm (an)}}
\def\hh{h^{(0)}}

\def\hhs#1{h^{(#1)}}
\def\ths#1{\theta^{(#1)}}
\def\vths#1{\vartheta^{(#1)}}

\def\whh{\widehat h}

\def\whg{\widehat g}

\def\whG{\widehat G}

\def\Thi#1{\Theta^{(#1)}}
\def\Gi#1{G^{(#1)}}
\def\Mi#1{M^{(#1)}}
\def\Di#1{\Delta^{#1}}
\def\whGi#1{\whG^{(#1)}}

\def\OmMSW{\Omega^{\rm MSW}}

\newcommand{\q}{\mbox{q}}
\def\DDb{D6-$\overline{\rm D6}$\ }

\def\mm{\ell_0}
\def\Nr{r}

\newcommand{\symfootnote}[1]{%
\let\oldthefootnote=\thefootnote%
\setcounter{mpfootnote}{2}%
\addtocounter{footnote}{-1}%
\renewcommand{\thefootnote}{\fnsymbol{mpfootnote}}%
\footnote{#1}%
\let\thefootnote=\oldthefootnote%
}

\title{Modular bootstrap for D4-D2-D0 indices on compact Calabi-Yau threefolds}

\preprint{arXiv:2204.02207v4}

\author{Sergei Alexandrov$^1$, Nava Gaddam$^2$, Jan Manschot$^3$,  Boris Pioline$^4$
\\
$^1$ {\it Laboratoire Charles Coulomb (L2C), Universit\'e de Montpellier,
CNRS, \\ F-34095, Montpellier, France}\\

$^2$ {\it Institute for Theoretical Physics and Center for Extreme Matter and Emergent Phenomena,
Utrecht University, 3508 TD Utrecht, The Netherlands}\\

$^3$ {\it School of Mathematics, Trinity College, Dublin 2, Ireland}\\

$^4$ {\it Laboratoire de Physique Th\'eorique et Hautes
Energies (LPTHE), UMR 7589 CNRS-Sorbonne Universit\'e,
Campus Pierre et Marie Curie, \\
4 place Jussieu, F-75005 Paris, France} \\

\vspace*{2mm} {\tt e-mail:
\email{sergey.alexandrov@umontpellier.fr},
\email{gaddam@uu.nl},
\email{manschot@maths.tcd.ie},
\email{pioline@lpthe.jussieu.fr}
}

\vspace*{-3mm}

}

\abstract{
We investigate the modularity constraints on the generating series $h_r(\tau)$ of BPS
indices counting D4-D2-D0 bound states with fixed D4-brane charge $r$ in type IIA
string theory compactified on complete intersection Calabi-Yau threefolds with $b_2 = 1$.
For unit D4-brane, $h_1$ transforms as a (vector-valued) modular form under the action of $SL(2,\IZ)$
and thus is completely determined by its polar terms. We propose an Ansatz for these terms in terms
of rank 1 Donaldson-Thomas invariants, which incorporates contributions from a single D6-$\overline{\rm D6}$ pair.
Using an explicit overcomplete basis of the relevant space of weakly holomorphic modular forms
(valid for any $r$), we find that for 10  of the 13 allowed threefolds,
the Ansatz leads to a solution for $h_1$ with integer Fourier coefficients,
thereby predicting an infinite series of DT invariants.
For $r > 1$, $h_r$ is mock modular and determined by its polar part together with its shadow.
Restricting to $r = 2$, we use the generating series of Hurwitz class numbers to construct
a series $\han_{2}$ with exactly the same modular anomaly as $h_2$,
so that the difference $h_{2}-\han_2$ is an ordinary modular form fixed by its polar terms.
For lack of a satisfactory Ansatz, we leave the determination of these polar terms as an open problem.
}

\begin{document}

\setlength{\parskip}{0.2cm}

\section{Introduction}

Elucidating the microscopic origin of the Bekenstein-Hawking entropy of black holes has been one of
the most fruitful endeavours in string theory, with amazing quantitative success for BPS black holes
in highly supersymmetric vacua. In trying to extend this program for type IIA string vacua with
$\cN=2$ supersymmetry in four dimensions (the minimal amount of supersymmetry allowing for BPS states),
D4-D2-D0 black holes play a special role, as they can be lifted to M5-branes wrapped
on divisor $\cD$ (i.e. a complex four-cycle) inside the Calabi-Yau (CY) threefold $\CY$  \cite{Maldacena:1997de}.
This allows an effective description in terms of a two-dimensional $(0,4)$ superconformal field theory (SCFT)
obtained by reducing the (still mysterious) six-dimensional $(2,0)$ SCFT on the five-brane world-volume \cite{Maldacena:1997de, Minasian:1999qn}.
In particular, the generating series of BPS indices for fixed D4 and D2-brane charges is determined
by the elliptic genus of the two-dimensional SCFT and is therefore expected to be modular.
This fact can be used to bootstrap\footnote{The term `bootstrap' covers a multitude of
non-perturbative approaches to determine physical quantities from general consistency constraints and some
basic assumptions about the spectrum. In the context of two-dimensional (super) conformal field theories,
constraints from modularity are a powerful tool to learn about the (BPS and non-BPS) spectrum, see for example in
\cite{Witten:2007kt, Gaberdiel:2008xb, Gaiotto:2007xh,Hellerman:2009bu, Keller:2012mr, Collier:2016cls, Hartman:2019pcd, Mukhi:2019xjy}. }
the full generating series from some small set of data, such as
BPS indices of D4-D2-D0 bound states with the smallest possible values of the D0-brane charge $q_0$.
From a mathematical viewpoint, this opens a way to access an infinite set of rank-zero Donaldson-Thomas (DT) invariants,
which are notoriously difficult to compute directly, and a straightforward method
to extract the asymptotic growth of the BPS indices in the `Cardy regime'
$q_0\to\infty$ with fixed D4 and D2-brane charges.

In practice, this approach requires (i) a complete characterization of the modular properties
of the generating series, and (ii) an ability to determine its polar coefficients.
So far it has been implemented only for a few examples of compact CY threefolds with one K\"ahler modulus
(such as the quintic threefold) and a single D4-brane wrapping a smooth ample divisor $\cD$.
For such divisors, the generating series $h_{1,\mu}(\tau)$ of BPS indices, with
fixed primitive D4-brane charge $p=[\cD]\in H_4(\CY)$, D2-brane charge $\mu\in H_2(\CY)$
(modulo spectral flow)
and fugacity $\q=e^{2\pi\I\tau}$ conjugate to the D0-brane charge,
behaves as a vector-valued (VV) modular form of fixed weight
$w=-\frac12 b_2(\CY)-1$ and fixed multiplier system under $SL(2,\IZ)$ transformations
$\tau\mapsto (a\tau+b)/(c\tau+d)$ \cite{Maldacena:1997de,Gaiotto:2006wm,deBoer:2006vg}.
Since the dimension of the space of such VV modular forms is bounded
by the number of {\it polar} terms \cite{Bantay:2007zz,Manschot:2007ha,Manschot:2008zb},
i.e. those with inverse powers of $\q$,
the knowledge of the latter is sufficient to fix the whole generating series.
The computation of $h_{1,\mu}(\tau)$ thus reduces to fixing a finite number of coefficients.

Several techniques  for computing the polar coefficients have been developed in the literature,
either by directly quantizing the moduli space of D4-D2-D0-brane
configurations for low D0-brane charge, as demonstrated for the quintic threefold
in \cite{Gaiotto:2005rp,Gaiotto:2006wm} and for a handful of other one-parameter models in \cite{Gaiotto:2007cd},
or by using the AdS/CFT correspondence in the near horizon
geometry \cite{Gaiotto:2006wm}.
A more systematic approach is to view the polar D4-D2-D0 states as bound states of D6-branes and
$\overline{\rm D6}$-branes \cite{Denef:2007vg}, and exploit the relation $Z_{D6}\sim Z_{\rm top}$
between the partition function $Z_{D6}$ of D6-D2-D0 indices with unit D6-brane charge
(also known as rank-one DT invariants) and the topological
string partition function $Z_{\rm top}$ \cite{gw-dt},
which is in turn determined by the Gopakumar-Vafa (GV) invariants.
This approach was applied in \cite{Collinucci:2008ht,VanHerck:2009ww}
to the quintic and a couple of other one-parameter CYs,
confirming the analysis in \cite{Gaiotto:2006wm,Gaiotto:2007cd}.
It however assumes that only a single \DDb pair
contributes, an assumption which needs to be verified case-by-case
through a detailed analysis of the possible attractor flow trees
\cite{Denef:2001xn,Denef:2007vg,Manschot:2010xp,Alexandrov:2018iao}
or more general\footnote{Multi-centered scaling solutions \cite{Denef:2002ru,Bena:2006kb,Bena:2012hf,Descombes:2021egc}
are not accessible by attractor flow tree techniques,
but are amenable to localization methods~\cite{Manschot:2010qz,Manschot:2011xc}.}
multi-centered configurations \cite{deBoer:2008zn,Manschot:2010qz,Manschot:2011xc,Gaddam:2016xum}.
An important check  is that the proposed polar terms should allow
for the existence of a VV modular form
with integer Fourier coefficients. This requirement is especially non-trivial when the dimension of
the space of VV modular forms is strictly smaller than the number of polar terms,
which implies that  the polar coefficients must satisfy certain
linear constraints for a modular form with this polar part to exist
\cite{Manschot:2007ha,Manschot:2008zb}.

\begin{table} [t]
\begin{centering}
$$
\begin{array}{|l|r|r|r|r|r|r|r|r|r|r|}
\hline \mbox{CICY} & \chi& \kappa  &c_{2}  & \chi(\cO_\cD)  &  \chi(\cO_{2\cD})   & n_1 & n_2 & C_1 & C_2  & {\rm type} \\   \hline
X_5(1^5)   &-200   &5                 &  50 &   5 & 15  & 7 &36 & 0 & 1 & F\\
 X_6(1^4,2)& -204& 3 & 42& 4 &11  & 4 &   19 & 0 & 1 &F \\
X_8(1^4,4)  &-296   &2              &  44  &   4 & 10&   4 &14 & 0 & 1 &F \\
 X_{10}(1^3,2,5)& -288&      1 & 34& 3 &7  & 2 & 7 &0 & 0 &F \\
X_{4,3}(1^5,2)&-156   &6              &  48 &    5 &16 &  9 & 42 & 0 & 0 &F \\
 X_{4,4}(1^4,2^2) & -144&          4 & 40& 4  &12 &  6 & 25 & 1 & 1 &K\\
X_{6,2}(1^5,3)&-256   &4             &  52  &    5 &14 &  7 & 30 & 0 & 1 &C\\
 X_{6,4}(1^3,2^2,3)& -156&         2 & 32& 3 &8 &  3 & 11 & 0 & 1 &F\\
X_{6,6}(1^2,2^2,3^2)&-120   &1      &  22 & 2&5 &  1 & 5 & 0 & 0 &K\\
X_{3,3}(1^6)&-144   &9         &  54   &   6 &21  & 14 & 78 & 1 & 3 &K \\
 X_{4,2}(1^6)& -176&   8 & 56& 6 &20 &  15 & 69 & 1 & 3 &C\\
X_{3,2,2}(1^7)&-144   &12              &  60 &    7 &26 &  21& 117 & 1 & 0 &C\\
 X_{2,2,2,2}(1^8)    & -128&       16 & 64 & 8 &32 & 33 & 185 & 3 & 4 & M\\
\hline
\end{array}
$$
\caption{Relevant data for the 13 smooth CICY threefolds with $b_2(\CY)=1$
(the first 4 columns are taken from \cite{Huang:2006hq}, \cite[Table 3.1]{Huang:2007sb}).
A complete intersection of multidegree $(d_1,\ldots, d_k)$ in weighted projective space
$\mathbb{P}^{k+3}(w_1,\ldots,w_l)$
is denoted $X_{d_1,\ldots, d_k}( w_1^{m_1},\ldots,w_p^{m_p})$ where $m_i$ is the number
of repetitions of the weight $w_p$. The columns $\chi, \kappa, c_2$ indicate the Euler number of $\CY$,
intersection product $\kappa=\cD\cup \cD \cup \cD$ and second Chern class $c_2=\int_{\cD} c_2(T\CY)$.
The columns  $\chi(\cO_\cD)$ and $\chi(\cO_{2\cD})$
indicate the holomorphic Euler characteristic of the structure sheaf on $\cD$ and $2\cD$,
which determine the central charge $c_R$ in the $(0,4)$ SCFT.
The columns $n_1,n_2$ indicate the number of polar terms in the generating series
$h_{1,\mu}$ and $h_{2,\mu}$ (assuming the Castelnuovo bound on GV invariants),
while $C_1,C_2$ indicate the difference between the number of polar terms and
the actual dimension of the space of VV modular forms. Each model has a one-dimensional K\"ahler moduli space with
three singular points, including the standard large volume limit at $\psi=\infty$
and a conifold singularity at $\psi=1$; the last column of the table indicates the type of singularity at $\psi=0$
in the terminology of \cite{Joshi:2019nzi}.
\label{table1}}
\end{centering}
\end{table}

\medskip

In this work, our aim is twofold. First, we generalize the analysis of
\cite{Gaiotto:2006wm,Gaiotto:2007cd,Collinucci:2008ht,VanHerck:2009ww}
to the complete list \cite{Doran:2005gu}\footnote{We do not include the 14-th case $X_{2,12}$
from \cite{Doran:2005gu}, since it does not correspond to a smooth threefold \cite{clingher201514th}.
For brevity we also refrain from considering other types of compact CY threefolds  with $b_2(\CY)=1$,
such as the $\IZ_5$ quotient of the quintic \cite{Gaiotto:2006wm},
R{\o}dland's determinantal and Pfaffian threefolds \cite{Hosono:2007vf},
or the Reye congruence threefold \cite{Hosono:2011np}, but our analysis applies just as well to these cases.
}
of 13 compact CY threefolds with $b_2(\CY)=1$ obtained as a complete intersection
in weighted projective space (CICY, see Table \ref{table1}).
Second, we attempt to extend this analysis to the case of non-primitive D4-brane charge $p=2[\cD]$.
Our main results in the first direction
are as follows (the first item being also relevant for the second goal):

\begin{itemize}
\item[a)] An explicit overcomplete basis spanning the space $\scM_{\Nr}(\CY)$ of VV modular forms
characterized by the weight and multiplier system of $h_{\Nr,\mu}$ where $\Nr$ is the wrapping number,
i.e $p=\Nr[\cD]$ with $[\cD]$ being the primitive generator of $H_4(\CY,\IZ)$.
This basis is similar in spirit to the one constructed for $\Nr=1$ in \cite{Gaiotto:2007cd},
but less contrived and valid for any $\Nr$ and all one-parameter threefolds
(the dependence on  $\CY$ arises through the triple intersection number $\kappa$ and the second Chern class $c_2$.)
We use it to check the difference between the dimension of $\scM_\Nr(\CY)$
and the number of polar terms predicted by the Selberg trace formula \cite{Manschot:2008zb},
and to reconstruct VV modular forms from their polar coefficients.

\item[b)]
An Ansatz \eqref{hpolar} for the polar terms of  $h_{1,\mu}$,
in terms of the rank-one Donaldson-Thomas invariants of $\CY$,
which sums the contributions from a single \DDb pair
and reproduces the results of \cite{Gaiotto:2007cd} for
$\CY\in \{X_5,X_6,X_8,X_{10},X_{3,3}\}$. A similar Ansatz is proposed for $h_{\Nr,\mu}$ with $\Nr>1$ in \eqref{hpolargenr},
but we expect that multiple \DDb pairs in general contribute,
so that Ansatz probably captures only part of the contributions to the polar coefficients.

\item[c)]For $\Nr=1$ and all but 3 of the 13 models considered, using the GV invariants computed
in \cite{Huang:2006hq}, we find that the polar coefficients predicted by our Ansatz are consistent
with the existence of a VV modular form with integer coefficients.
Notably, this includes two cases ($X_{3,3}$, already considered in  \cite{Gaiotto:2007cd},
and $X_{4,4}$) where the dimension of the space of VV modular forms is one less than
the number of polar terms, providing a rather strong check on the validity of the Ansatz \eqref{hpolar}.

\item[d)] For the remaining 3 models $X_{4,2},X_{3,2,2},X_{2,2,2,2}$,
we find that the polar coefficients predicted by the Ansatz do not allow for a VV modular form
with integer coefficients. Barring possible errors in the  tables of GV invariants
in \cite{Huang:2006hq}, we suspect that for such models there are additional contributions
to the polar terms which we have not identified.

\end{itemize}

Our second goal is to propose a procedure to construct the generating series $h_{p,\mu}(\tau)$ for higher D4-charges $p=\Nr[\cD]$,
and make it explicit in the case $\Nr=2$.
The main difficulty, which  appears when $\cD$ is a reducible divisor or when the D4-brane charge
is a multiple of an irreducible divisor class, is that  $h_{p,\mu}$ fails to be a modular form.
Instead, it turns out to be a VV {\it mock} modular form of {\it higher depth} and {\it mixed type},
which transforms inhomogeneously under the action of $SL(2,\IZ)$
\cite{Alexandrov:2016tnf,Alexandrov:2017qhn,Alexandrov:2018lgp,Chattopadhyaya:2021rdi}
(following on earlier works \cite{Manschot:2009ia,Alim:2010cf,Dabholkar:2012nd,Cheng:2017dlj}).
More specifically, this means that $h_{p,\mu}$ admits a
non-holomorphic completion $\widehat{h}_{p,\mu}$
that transforms as a VV modular form of the same weight and
multiplier system as in the irreducible case,
at the expense of being non-holomorphic. In \cite{Alexandrov:2016tnf,Alexandrov:2018lgp,Alexandrov:2019rth},
this completion has been constructed explicitly in terms of products of $h_{p_i,\mu_i}$'s such
that $p=\sum_{i=1}^\Nr p_i$ (with $\Nr-1$ being the {\it depth}). The fact that $h_{p,\mu}$
is of {\it mixed type} (meaning that the $\bar\tau$-derivative of
$\widehat{h}_{p,\mu}$, known as the shadow, is not anti-holomorphic)
implies that it is necessary to specify both the polar terms and the
shadow\footnote{In contrast, `pure' mock modular forms can be recovered
from their polar terms by a Poincar\'e series-type construction,
which produces both the $\q$-expansion and the shadow, see e.g. \cite{Manschot:2007ha,Cheng:2012qc}.}
in order to fix $h_{p,\mu}$ uniquely.
Given this information, $h_{p,\mu}$ can be reconstructed in two steps.
First, one produces an {\it ad hoc} function $\han_{p,\mu}$ with the same modular anomaly as $h_{p,\mu}$,
such that $\hh_{p,\mu}=h_{p,\mu}-\han_{p,\mu}$ becomes an ordinary VV holomorphic modular form.
Second, one reconstructs  $\hh_{p,\mu}$ from its polar coefficients
by expanding on an explicit (possibly overcomplete) basis, thereby obtaining the generating series $h_{p,\mu}$ of interest.
Clearly, the second step is straightforward (with the help of a computer), while the first step requires some ingenuity.

In this work we implement this idea for the generating series $h_{2,\mu}$
of D4-D2-D0 indices with D4-brane charge $p=2[\cD]$.
Namely, we construct a VV mock modular form $\han_{2,\mu}$ with the same
modular anomaly as $h_{2,\mu}$,
by acting with a suitable Hecke operator $\cT_{\kappa}$ \cite{Bouchard:2018pem}
on the generating series of Hurwitz class numbers  (the simplest example of a VV mock
modular form of depth one).\footnote{Our construction assumes that $\kappa$ is a power of a prime number,
which is the case for all models in Table \ref{table1} except $X_{4,3}$ and $X_{3,2,2}$.
We leave it as an open problem to extend it to general $\kappa$.}
The latter are well-known to arise as rank $2$ Vafa-Witten invariants  on the complex projective plane
\cite{Vafa:1994tf}, which also count D4-branes wrapped twice on the compact divisor
$\IP^2$ inside the non-compact threefold $K_{\IP^2}$.
We then provide an explicit algorithm that determines $h_{2,\mu}$, assuming that its polar coefficients are known.
Unfortunately, when applied to the polar coefficients stipulated by the Ansatz \eqref{hpolargenr},
it fails to produce satisfactory results:
either the polar coefficients do not satisfy the constraints imposed by modularity, or
the resulting generating series turns out to have non-integer Fourier coefficients.
However, as emphasized above,  the Ansatz is unlikely to be correct when $\Nr>1$ anyway.

A detailed supergravity analysis of the multi-centered D4-D2-D0 bound states contributing
to the polar terms is left for future work. Until then, the generating series of
DT invariants computed by our method should be considered as tentative.
We note however that the idea that rank-zero DT invariants (counting D4-D2-D0 bound states)
are determined by rank-one invariants (counting D6-D4-D2-D0 bound states with unit D6-brane charge)
is broadly consistent with the OSV conjecture \cite{Ooguri:2004zv}
and with recent results in the mathematics literature
\cite{Toda:2011aa,Feyzbakhsh:2021rcv,Feyzbakhsh:2021nds,Feyzbakhsh:2022ydn},
although the detailed connection remains elusive.

\medskip

The remainder of this article is  organized as follows. In \S\ref{sec-DT} we briefly review the
definition of ordinary DT invariants,
 D4-D2-D0 indices, and the modular constraints that their generating series
$h_{\Nr,\mu}$ ought to satisfy, specializing to the one-modulus case.
In \S\ref{sec-vecmod} we construct an overcomplete basis of the space $\scM_\Nr(\CY)$
of vector-valued modular forms
in which these generating series would live if the modular anomaly were absent.
In \S\ref{sec-DT1}, we consider D4-D2-D0 indices with unit D4-brane charge $\Nr=1$,
propose an Ansatz for the polar part of the corresponding generating series $h_{1,\mu}$,
and determine the corresponding modular forms. In \S\ref{sec-DT2}, we turn to the $\Nr=2$ case,
and develop a strategy for determining the VV mock modular forms $h_{2,\mu}$, assuming their polar part is known.
In \S\ref{sec-disc} we discuss the possible origin of additional contributions to the polar terms.
In Appendix \ref{sec-constraints} we derive an explicit formula for the dimension of the space
$\scM_{\Nr}(\CY)$, and tabulate the results for low values of $r$.
In \S\ref{sec-Hecke} we construct a Hecke operator producing a solution $\han_{2,\mu}$ of the modular anomaly equation
from the generating series of Hurwitz class numbers.
In \S\ref{sec-genfun1}, for each of the 13 one-parameter CICY threefolds,
we provide the rank 1 DT invariants and the resulting VV modular forms $h_{1,\mu}$ together with their $q$-expansions.
Finally in \S\ref{sec_math}, we review some recent results in the mathematical literature on
rank 0 DT invariants, and compare them with our Ansatz for polar terms.

\medskip

\noindent {\it Note added in v2, updated in v3:} After the first release of this work on {\tt arXiv},
it became apparent in discussions with Soheyla Feyzbakhsh that
the polar coefficients of the generating series $h_{N,\mu}$ can be computed  by generalizing
the approach of \cite{Feyzbakhsh:2022ydn}, provided GV invariants are known to sufficiently high genus and degree. Recent work  using  this strategy \cite{followup} confirms that the
generating series quoted in Appendix \ref{sec-genfun1} are indeed correct for $X_5, X_6, X_8, X_{3,3}, X_{4,4}, X_{6,6}$, verifying the Ansatz \eqref{hpolar} in those cases. However, it is found that some of the polar terms for $X_{10},X_{6,2}, X_{6,4}, X_{4,3}$ and $X_{4,2}$
disagree with this Ansatz. In particular, for $X_{10}$,
the suggestion of \cite{VanHerck:2009ww}, which was initially dismissed in the first release of the present work,
is in fact confirmed (see footnote  \ref{fooWyder}). The polar terms for
$X_{3,2,2},X_{2,2,2,2}$ are currently out of reach by this approach,
due to limitations in computing
GV invariants. We chose to leave unchanged the results in Appendix \ref{sec-genfun1},
but to mark with $\dagger$ the results that we now believe to be incorrect.
\label{noteadded}

\section{DT invariants and D4-D2-D0 bound states}
\label{sec-DT}

In this section we provide a lightning review of BPS indices counting supersymmetric
D4-D2-D0 bound states in type II string theory compactified on a CY threefold $\CY$, and of the
the modular properties of their generating series $h_{p,\mu}$, specializing the relevant formulae
to the one-modulus case $b_2(\CY)=1$.
We refer the reader to our previous works \cite{Alexandrov:2016tnf,Alexandrov:2018lgp} for more details.

\subsection{Generalized DT invariants and spectral flow}
\label{subsec-DTgen}

Recall that in the large volume limit, D6-D4-D2-D0 bound states on $\CY$ are described by
semi-stable coherent sheaves $\cE$ on $\CY$. Their electromagnetic charge $\gamma$
is identified with the Mukai vector $\ch \cE \sqrt{\Td \CY}$.
Expanding $\gamma$ on a basis of $H^{\rm even}(\CY,\IZ)$,
we obtain components $\gamma=(p^0,p^a,q_a,q_0)$ with $a=1,\dots, b_2$ satisfying
the following quantization conditions \cite{Alexandrov:2010ca}:
\be
\label{fractionalshiftsD5}
p^0, p^a\in\IZ ,
\qquad
q_a \in \IZ  + \frac12 \,\kappa_{abc}p^b p^c - \frac{1}{24} p^0 c_{2,a} ,
\qquad
q_0\in \IZ-\frac{1}{24}\,c_{2,a} p^a,
\ee
where $\kappa_{abc}$ are triple intersection numbers of $\CY$ and $c_{2,a}$ are components of its second Chern class.
The mass of such BPS states is proportional to the modulus of the central charge,
which is given in the large volume limit by
\be
Z_\gamma = \int_{\CY} e^{-  z^a \omega_a} \ch \cE \sqrt{\Td \CY},
\ee
where $z^a=b^a+\I t^a$ are the K\"ahler moduli conjugate to the basis $\omega_a$ in $H^2(\CY,\IZ)$.
Under a large gauge transformation $b^a\to b^a+\epsilon^a$, the central charge and hence the mass stay invariant provided
$\cE$ is tensored with a line bundle $\cF$ with $c_1(\cF) = \epsilon^a\omega_a$, an operation
known as spectral flow which shifts the charges as follows
\be
\label{specflowD6}
\begin{split}
p^0\mapsto p^0,
\qquad &
p^a\mapsto p^a+\epsilon^a p^0,
\qquad
q_a \mapsto q_a - \kappa_{abc} p^b \epsilon^c -\frac{p^0}{2}\, \kappa_{abc} \epsilon^b \eps^c,
\\
q_0 \mapsto&\, q_0 -\eps^a q_a +\frac12\, \kappa_{abc} p^a \eps^b \eps^c +
\frac{p^0}{6}\, \kappa_{abc}  \eps^a  \eps^b \eps^c.
\end{split}
\ee
We will denote the resulting charge by $\gamma[\eps]$.

The BPS index $\Omega(\gamma;z^a)$, or generalized Donaldson-Thomas (DT) invariant,
is defined (informally) as the signed Euler number of the
moduli space of semi-stable sheaves with fixed charge $\gamma$, where semi-stability requires that
all subsheaves $\cE'\subset \cE$ have $\arg Z_{\gamma'}\leq \arg Z_{\gamma}$.
Rational DT invariants are defined by the usual multicover formula
\be
\label{defntilde}
\bar\Omega(\gamma;z^a) = \sum_{d|\gamma}  \frac{1}{d^2}\,
\Omega(\gamma/d;z^a)\, ,
\ee
so that $\bar\Omega(\gamma;z^a)=\Omega(\gamma;z^a)$ whenever the charge $\gamma$ is primitive.
Both $\Omega(\gamma;z^a)$ and its rational counterpart are invariant under the
spectral flow \eqref{specflowD6} provided it is combined with $b^a\to b^a+\epsilon^a$.

\subsection{Rank 1 DT invariants and GV invariants}
\label{sec_GVDT}

While our primary interest is in D4-D2-D0 bound states, an important ingredient will be
the ordinary DT invariants which count D6-D4-D2-D0 bound
states with a single unit of D6-brane charge, in the large volume limit. Due to the
symmetry \eqref{specflowD6}, they may be expressed in terms of the invariant D2 and D0 charges\footnote{The shift
proportional to the second Chern class ensures that $Q_a$ is integer,
whereas the integrality of $n$ follows from the integrality of the arithmetic genus $\chi(\cO_{\cD_p})$
in \eqref{defL0}. }
\be
\label{defQn}
Q_a = q_a + \frac12\, \kappa_{abc} p^b p^c + \frac{c_{2a}}{24}\, ,
\qquad
n = - q_0 - p_a q^a - \frac13\, \kappa_{abc} p^a p^b p^c\, .
\ee
Following \cite[\S6.1.2]{Denef:2007vg}, we denote the corresponding DT invariants by
\be
DT(Q_a,n) = \lim_{\lambda\to+\infty} \Omega(1,p^a,q_a,q_0;\lambda(b^a+\I t^a)),
\label{defDTa}
\ee
where $b^a< -t^a \sqrt3$. If instead $b^a>  t^a \sqrt3$, one has
\be
DT(Q_a,-n) = \lim_{\lambda\to+\infty} \Omega(1,p^a,q_a,q_0;\lambda(b^a+\I t^a)).
\label{defDTb}
\ee
Since $DT(Q_a,n)$ vanishes for $n$ negative  and large enough (as a result of Castelnuovo-type bounds), one can construct
the formal series
\be
\label{defZDT}
Z_{DT} (\xi^a,q)= \sum_{Q_a,n} DT(Q_a,n)\,
e^{2 \pi \I Q_a \xi^a}  q^{n},
\ee
where the sum runs over effective curve classes $Q_a\geq 0$. Up to a factor $M(-q)^{\chi_{\scriptstyle\CY}}$,
where $M(q)=\prod_{k>0}(1-q^k)^{-k}$ is the Mac-Mahon function,
the series \eqref{defZDT} coincides with the generating function of stable pair invariants
(see e.g. \cite{Pandharipande:2011jz}). More importantly for our purposes, the series \eqref{defZDT}
can in turn be expressed in terms of GV
invariants $N_Q^{(g)}$ using the GV/DT correspondence \cite{gw-dt,gw-dt2}
\be
\begin{split}
\label{gvc2}
Z_{DT}(\xi^a,q)=&\,
[M(-q)]^{\chi_{\scriptstyle\CY}}
\prod_{Q>0}\prod_{k>0} \left(1-(-q)^k e^{2\pi\I Q_a \xi^a}\right)^{k N_Q^{(0)}}
\\
&\times
\prod_{Q>0}\prod_{g>0}
\prod_{\ell=0}^{2g-2}
\left(1- (-q)^{g-\ell-1} e^{2\pi\I Q_a \xi^a}
\right)^{(-1)^{g+\ell} {\scriptsize \begin{pmatrix} 2g-2 \\ \ell \end{pmatrix}}
N_Q^{(g)}} .
\end{split}
\ee
The right-hand side is well-defined as a formal series, since for any fixed $Q_a$
there is only a finite number of $g$ such that $N_Q^{(g)}\neq 0$. Upon setting $q=e^{\I\lambda}$ and expanding
as $\lambda\to 0$, it provides the perturbative expansion of the topological string partition function,
which can in principle be computed by solving the holomorphic anomaly equations (see e.g. \cite{Klemm:2004km}).
Thus, \eqref{gvc2} gives a practical way of computing the rank 1 DT invariants $DT(Q_a,n)$.

\subsection{Rank 0 DT invariants and their generating series}
\label{subsec-MSWgen}

We now turn to our prime interest, namely
D4-D2-D0 bound states with vanishing  D6-brane charge, $p^0=0$.
In this case, the D4-brane charge $p^a$ is invariant under spectral flow, along with
 the following combination of D0 and D2 charges\footnote{Note that our definition of $\hat q_0$ differs from \cite{Maldacena:1997de}
by an overall sign.}
\be
\label{defqhat}
\hat q_0 \equiv
q_0 -\frac12\, \kappa^{ab} q_a q_b .
\ee
Here $\kappa^{ab}$ is the inverse of $\kappa_{ab}=\kappa_{abc} p^c$, a quadratic form
of signature $(1,b_2(\CY)-1)$ on $\Lambda\otimes \IR\simeq \IR^{b_2(\CY)}$
where $\Lambda=H_4(\CY,\IZ)$.
The Bogomolov-Gieseker bound implies that the BPS index $\Omega(\gamma;z^a)$ vanishes
unless  the invariant charge $\hat q_0$ is bounded from above by
\be
\hat q_0 \leq \hat q_0^{\rm max}=\frac{1}{24}\, \chi(\cD_p),
\label{qmax}
\ee
where $\chi(\cD_p)$ is the Euler number of the divisor $\cD_p=p^a\gamma_a$
(with $\gamma_a$ a basis of effective divisor classes in $H_4(\CY,\IZ)$), given by \cite[Eq.(3.3)]{Maldacena:1997de}
\be
\chi(\cD_p) = \kappa_{abc} p^a p^b p^c+c_{2,a}p^a.
\label{defchiD}
\ee

Using the spectral flow \eqref{specflowD6}, one may remove most of the D2-brane charge $q_a$,
though  not all of it in general. More precisely, one can always decompose
\be
\label{defmu}
q_a = \mu_a + \frac12\, \kappa_{abc} p^b p^c + \kappa_{abc} p^b \epsilon^c,
\ee
for some $\epsilon^a\in\Lambda$ (which can be removed by spectral flow) and
$\mu_a\in \Lambda^*/\Lambda$ (which is invariant under spectral flow), where we use
the quadratic form $\kappa^{ab}$ to identify $\Lambda$ with its image in $\Lambda^*$.
The representative $\mu_a$ in the discriminant group $\Lambda^*/\Lambda$
(a finite group of order $|\det\kappa_{ab}|$) is sometimes known as the residual D2-brane charge.

When $p^a$ is irreducible, there are no walls of marginal stability in the large volume limit,
and the index $\Omega(\gamma;z^a)$ (equal to the rational DT invariant) is independent of $b^a$
and invariant under spectral flow. In contrast, when $p^a$ is reducible, there are walls
of marginal stability extending to large $t^a$, and in this regime $\Omega(\gamma;z^a)$
is only a locally constant function of $z^a$. We define the `MSW invariants'
$\OmMSW(\gamma)=\Omega(\gamma;z^a_\infty(\gamma))$ as the DT invariants evaluated at the large volume attractor point \cite{deBoer:2008fk},
\be
\label{lvolatt}
z^a_\infty(\gamma)=\lim_{\lambda\to +\infty}\(b^a(\gamma)+\I\lambda t^a(\gamma)\)
= \lim_{\lambda\to +\infty}\(-\kappa^{ab}q_b+\I\lambda  p^a\).
\ee
The MSW index $\OmMSW(\gamma)$ should be distinguished from the attractor index
$\Omega_\star(\gamma)$, though both are  by construction moduli-independent.
Since $\Omega(\gamma;z^a)$ is invariant under the combined action of the spectral flow
\eqref{specflowD6} and $b^a\to b^a+\epsilon^a$,
$\OmMSW(\gamma)$ is invariant under \eqref{specflowD6} itself, and therefore depends only
on $p^a,\mu_a$ and $\hat q_0$, so we denote it by $\OmMSW(\gamma)=\Omega_{p,\mu}( \hat q_0)$.
Setting $p=r p_0$ such that $p_0$ is primitive,  $\Omega_{p,\mu}( \hat q_0)$
is given informally by the signed Euler number of the combined moduli space
of the divisor $\cD_{p_0}$ inside $\CY$, equipped with a stable coherent sheaf $\cE$ of rank $r$,
slope $\mu/r$ and discriminant $\Delta=\frac{\chi(\cD_{rp_0})}{24 r} - \frac{\hat q_0}{r}$.
In particular, it is invariant under $\mu\mapsto -\mu$, corresponding to dualizing the sheaf $\cE$.

We can now define $h_{p,\mu}(\tau)$ as the generating series of rational 
MSW invariants\footnote{The definition  in terms rational MSW invariants
was proposed in \cite{Manschot:2009ia,Manschot:2010xp}, motivated by consistency with
wall-crossing. 
Note that the modular parameter
$\q=\exp(2\pi\I \tau)$ in \eqref{defhDT} is unrelated to $q$ in \eqref{defZDT},
which was not expected to have modular properties. The series $h_{p,\mu}(\tau)$
is invariant under $\mu\mapsto\mu+\epsilon$ with $\epsilon\in \Lambda$
and under $\mu\mapsto -\mu$.}
\be
h_{p,\mu}(\tau) =\sum_{\hat q_0 \leq \hat q_0^{\rm max}}
\bOm_{p,\mu}(\hat q_0)\,\q^{-\hat q_0 }\ .
\label{defhDT}
\ee
 As briefly explained in the Introduction,
the generating functions $h_{p,\mu}$ possess remarkable modular properties
under the standard $SL(2,\IZ)$ transformations $\tau\mapsto \frac{a\tau+b}{c\tau+d}$.
The precise properties depend on the divisor $\cD_p$ corresponding to D4-brane charge $p^a$.
If the divisor is irreducible, various physical arguments show\footnote{Even in this simple case,
modularity remains conjectural from a mathematical viewpoint, see e.g. \cite{Feyzbakhsh:2020wvm} for some recent discussion.}
\cite{Gaiotto:2006wm,deBoer:2006vg,Denef:2007vg,Alexandrov:2012au}
that $h_{p,\mu}$ is a weakly holomorphic VV modular form of weight $-\frac{b_2(\CY)}{2}-1$
with the multiplier system determined by the following two matrices for T and S-transformations
\cite[Eq.(2.10)]{Alexandrov:2019rth} (see also \cite{Gaiotto:2006wm, deBoer:2006vg, Denef:2007vg, Manschot:2007ha})
\be
\begin{split}
M_{\mu\nu}(T)=&\, e^{\pi\I\(\mu+\frac{p}{2}\)^2+\tfrac{\pi\I}{12}\, c_{2,a}p^a}\,\delta_{\mu\nu},
\\
M_{\mu\nu}(S)=&\, \frac{(-1)^{\chi(\cO_{\cD_p})}}{\sqrt{|\Lambda^*/\Lambda|}}\, e^{(b_2-2)\frac{\pi\I}{4}}\,
e^{-2\pi\I \mu\cdot\nu}\,,
\end{split}
\label{Multsys-hp}
\ee
where $\mu\cdot\nu=\kappa^{ab}\mu_a \nu_b$,
$\delta_{\mu\nu}$ is the Kronecker delta on the discriminant group $\Lambda^*/\Lambda$, and
$\chi(\cO_{\cD_p})=\frac12(b_2^+(\cD_p)+1)$ is the arithmetic genus given by
\be
\label{defL0}
\chi(\cO_{\cD_p}) = \frac16\, \kappa_{abc} p^a p^b p^c + \frac{1}{12}\, c_{2,a} p^a\, .
\ee

However, if the divisor can be decomposed into a sum of $\Nr$ irreducible divisors,
the generating function can be shown (using physical reasoning based on S-duality of Type IIB string theory \cite{Alexandrov:2018lgp})
to transform as a VV {\it mock} modular form of depth $\Nr-1$ (of the same weight and multiplier system as above).
This implies that its non-holomorphic completion
$\whh_{p,\mu}(\tau,\btau)$,
that transforms as a true modular form, is determined by iterated integrals of depth $\Nr-1$ of another modular form.
Although in \cite{Alexandrov:2018lgp} this modular completion has been found explicitly,
we do not need it here in full generality and will restrict to the case $\Nr=2$, first analyzed in \cite{Alexandrov:2016tnf}.
But before specifying its explicit form, let us further restrict to  CY threefolds
with just one K\"ahler modulus, the class that we analyze in this paper.

\subsection{One modulus case }
\label{subsec-CYone}

Upon restricting to CY threefolds with $b_2(\CY)=1$, many of the equations above  simplify.
Firstly, the indices $a,b,c\dots$ take a single value so that the D4-brane charge $p^a$,
residual D2-brane charge $\mu_a$, intersection numbers $\kappa_{abc}$ and second Chern class $c_{2,a}$ become
scalar quantities which we denote simply as $p$, $\mu$, $\kappa$ and $c_2$.
The discriminant group $\Lambda^*/\Lambda$
coincides with the cyclic group $\IZ_{\kappa p}$ so that $\mu$ can be taken to lie in the interval
$\{0,\dots,\kappa p-1\}$.

Denoting by $\cD$ the generator of $H_4(\CY,\IZ)$, we have $\cD_p=p\cD$.
Therefore, the degree of reducibility $\Nr$ of the divisor coincides with the corresponding D4-brane charge, $\Nr=p$.
The modular weight of the generating functions $h_{\Nr,\mu}$ is always $-3/2$ and
the multiplier system \eqref{Multsys-hp} reduces to
\be
\begin{split}
M_{\mu \nu}(T)=&\,
e^{\frac{\pi\I}{\kappa\Nr}(\mu+\frac{\kappa}{2} \Nr^2 )^2 +\frac{\pi\I}{12}\,\Nr c_{2} }
\,\delta_{\mu\nu},
\\
M_{\mu\nu}(S)=&\,
\frac{(-1)^{\cI_\Nr}}{\sqrt{\I \kappa\Nr}}\,
e^{-\frac{2\pi\I}{\kappa\Nr}\,\mu \nu},
\end{split}
\label{multsys-h2}
\ee
where
\be
\cI_\Nr=\chi(\cO_{\Nr\cD})=\frac{1}{6}\,\kappa \Nr^3+\frac{1}{12}\, c_2 \Nr.
\label{chiODN}
\ee

For $\Nr>1$ the generating functions $h_{\Nr,\mu}$ no longer transform as VV modular forms under $SL(2,\IZ)$,
but rather as  mock modular forms of depth $\Nr-1$ and mixed type.
For $\Nr=2$, their completion  $\whh_{2,\mu}$ can be deduced by specializing
Eq.(1.3) in \cite{Alexandrov:2016tnf} to the case $b_2(\CY)=1$.
This gives
\be
\whh_{2,\mu}(\tau,\btau)=h_{2,\mu}(\tau)+\sum_{\mu_1,\mu_2=0}^{\kappa-1} R_{\mu,\mu_1\mu_2}(\tau,\btau)
\, h_{1,\mu_1}(\tau)\, h_{1,\mu_2}(\tau),
\label{whh2}
\ee
where
\be
R_{\mu,\mu_1\mu_2}(\tau,\btau)=\delta^{(\kappa)}_{\mu_1+\mu_2-\mu}(-1)^{\mu'} \Thi{\kappa}_{\mu'}(\tau,\btau),
\label{defR}
\ee
with $\mu'=\mu-2\mu_1+\kappa$. Here $\delta^{(n)}_x$ is the mod-$n$ Kronecker delta defined by
\be
\label{defdelta}
\delta^{(n)}_x=\left\{ \begin{array}{ll}
1\  & \mbox{if } x=0\!\!\!\mod n,
\\
0\ & \mbox{otherwise},
\end{array}\right.
\ee
while $\Thi{\kappa}_{\mu}(\tau,\btau)$ is the non-holomorphic theta series
\be
\Thi{\kappa}_{\mu}(\tau,\btau)=\frac{1}{8\pi}\, \sum_{k\in 2\kappa\IZ+\mu}|k|\,\beta_{\frac{3}{2}}\!\(\frac{\tau_2}{\kappa}\, k^2\)
e^{-\frac{\pi\I \tau}{2\kappa}\, k^2 },
\label{defRmu}
\ee
where $\beta_{\frac{3}{2}}(x^2)=2|x|^{-1}e^{-\pi x^2}-2\pi \Erfc(\sqrt{\pi} |x|)$.
In particular, \eqref{defRmu} satisfies the holomorphic anomaly equation
\be
\p_{\btau}\Thi{\kappa}_{\mu}=\frac{\sqrt{\kappa}}{16\pi\I\tau_2^{3/2}}\sum_{k\in 2\kappa\IZ+\mu}
e^{-\frac{\pi\I \btau}{2\kappa}\, k^2 }.
\label{holanom-Th}
\ee

In the following sections we shall apply these structural results
for the generating functions $h_{\Nr,\mu}$ to the set of one-parameter CY threefolds
that can be obtained as complete intersections in weighted projective spaces.
The relevant topological data for the corresponding 13 models are specified in Table \ref{table1}.

\section{The space of vector-valued modular forms}
\label{sec-vecmod}

In this section, we analyze the space $\scM_\Nr(\CY)$ of weakly holomorphic vector-valued modular forms
transforming with the same weight (namely, $-3/2$) and multiplier system as the
generating function $h_{\Nr,\mu}$, As explained in the previous section, for $\Nr=1$
the generating series  $h_{1,\mu}$ belongs to $\scM_1(\CY)$, while for $\Nr>1$ the
modular anomaly of $h_{\Nr,\mu}$ only specifies it up to an element in $\scM_\Nr(\CY)$.
Thus, the results in this section will be relevant for both cases.

\subsection{Modular constraints on polar terms}

It is well known that any weakly holomorphic modular form $f_\mu(\tau)$ of weight $w<0$ is completely fixed by its polar part,
i.e. the part of its Fourier expansion
\be
f_\mu(\tau)=\sum_{n=0}^\infty c_{\mu}(n) \,\q^{n-\Delta_\mu}
\label{fourier}
\ee
that becomes singular in the limit $\tau\to\I\infty$ \cite{Manschot:2007ha}.
It is captured by the terms with $n<\Delta_\mu$ and the corresponding $c_{\mu}(n)$ are called `polar coefficients'.
The remaining coefficients are then uniquely determined, for example by constructing
a Poincar\'e series seeded by the polar terms.

Crucially however, the dimension of the space of modular forms is often smaller (though never larger)
than the number of polar terms, which means that the polar coefficients cannot be chosen completely at will.
To allow for the existence of a modular form with given polar part (as opposed to a mock modular form),
the polar coefficients must satisfy $n$ constraints where $n$ is the dimension of the space of cusp modular forms
of weight $2-w$. The latter can be computed, for example, using the Selberg trace formula \cite{Skoruppa85,Manschot:2008zb}.

In Appendix \ref{sec-constraints} we derive the number of polar terms and the number of constraints
that they must satisfy for the case relevant to our study, namely, VV modular forms of weight $-3/2$,
multiplier system \eqref{multsys-h2} and exponents (cf. \eqref{qmax})
\be
\Di{h}_\mu = \frac{\chi(\Nr\cD)}{24} - {\rm Fr}\[ \frac{\mu^2}{2\kappa \Nr} + \frac{\Nr \mu}{2} \] .
\label{Deltamu}
\ee
Here ${\rm Fr}(x)$ denotes the fractional part  $x - \lfloor x \rfloor$ and
\be
\chi(\Nr\cD) = \kappa \Nr^3+c_{2}\Nr.
\label{defchiD-one}
\ee
Applying these results to the 13 one-parameter CICYs, one finds
the data provided in the last four columns of Table \ref{table1}.

\subsection{A universal basis}

For our purposes, we will need a (overcomplete) basis in $\scM_\Nr(\CY)$, which is the space of vector-valued modular forms
of weight $-3/2$, multiplier system \eqref{multsys-h2} and exponents
$\Delta_\mu$ specified in \eqref{Deltamu}.
A convenient choice of a basis can be constructed using the following set of theta series
\be
\label{Vignerasth}
\vths{m,p}_{\mu}(\tau,z)=
\!\!\!\!
\sum_{{k}\in \IZ+\frac{\mu}{m}+\frac{p}{2}}\!\!
(-1)^{mpk}\, \q^{\tfrac{m}{2}\,k^2} e^{2\pi\I m k z}.
\ee
They satisfy
\be
\vths{m,p}_{\mu}(\tau,z)= \vths{m,p}_{\mu+m}(\tau,z)= \vths{m,p}_{-\mu}(\tau,z)
\ee
and transform under $(\tau,z)\mapsto (\frac{a\tau+b}{c\tau+d}, \frac{z}{c\tau+d})$
as a vector-valued Jacobi form of weight 1/2
and multiplier system given by
\be
\begin{split}
M^{(m,p)}_{\mu\nu}(T)=&\, e^{\frac{\pi\I}{m} \(\mu+\tfrac{mp}{2} \)^2}\,\delta_{\mu\nu}\, ,
\\
M^{(m,p)}_{\mu\nu}(S)=&\,
\frac{e^{-\frac{\pi\I}{2}\, m p^2}}{\sqrt{\I m}}\,
e^{-2\pi\I\,\frac{\mu\nu}{m}}.
\end{split}
\label{eq:thetatransforms}
\ee
Note that $\vths{1,1}_{0}(\tau,z)$ coincides with the ordinary Jacobi theta series $\vth_1(\tau,z)$.
Let us then set $(m,p)=(\kappa \Nr,\Nr)$ and consider ratios of the form
\be
\frac{\ths{\Nr,\kappa}_{\mu}(\tau)}{\eta^{4\kappa \Nr^3 +\Nr c_2}(\tau)},
\qquad
\ths{\Nr,\kappa}_{\mu}(\tau)=\left\{
\begin{array}{ll}
\vths{\kappa\Nr,\Nr}_{\mu}(\tau,0), & \quad \kappa r \mbox{ even},
\\
-\frac{1}{2\pi \Nr}\,\p_z\vths{\kappa\Nr,\Nr}_{\mu}(\tau,0), & \quad \kappa r \mbox{ odd},
\end{array}
\right.
\label{defxiforms}
\ee
where $\eta(\tau)$ is the Dedekind eta function.
These functions are modular forms of weight
$-\hf(4\kappa \Nr^3 +\Nr c_2-1)+\delta^{(2)}_{\kappa r-1}$.
Taking into account that the multiplier system of the Dedekind eta function is given by
\be
M^{(\eta)}(T)=e^{\frac{\pi\I}{12} },
\qquad
M^{(\eta)}(S)=e^{-\frac{\pi\I}{4}},
\label{multsys-eta}
\ee
it is straightforward to check that the multiplier system of \eqref{defxiforms} coincides with \eqref{multsys-h2}.
Furthermore, given that $\eta(\tau)\sim \q^{1/24}$ as $\tau\to\I\infty$, it is easy to see that they have the Fourier expansion
of the form \eqref{fourier} with
\be
\Delta_\mu=\frac{4\kappa \Nr^3+c_2\Nr}{24}-\frac{\kappa\Nr}{2}\({\rm Fr'}\[\frac{\mu}{\kappa\Nr}+\frac{r}{2} \]\)^2
=\Di{h}_\mu+n,
\ee
where ${\rm Fr'}(x) = |x - \lfloor x \rceil| $ is the difference with the closest integer,
$\Di{h}_\mu$ is defined in \eqref{Deltamu}, and $n\in \IZ$.
Importantly, the integer $n$ is non-negative,
\be
n=\frac{\kappa \Nr^3}{8}-\frac{\kappa\Nr}{2}\({\rm Fr'}\[\frac{\mu}{\kappa\Nr}+\frac{r}{2} \]\)^2
+{\rm Fr}\[ \frac{\mu^2}{2\kappa \Nr} + \frac{\Nr \mu}{2} \]
\ge \frac{\kappa\Nr}{8}\, (\Nr^2-1)\ge 0,
\ee
where we used that ${\rm Fr'}(x)\le \hf$ and ${\rm Fr}(x)\ge 0$.
This allows to conclude that
the functions \eqref{defxiforms} have the same or larger number of polar terms as we need.

These considerations motivate us to introduce the functions
\be
\hhs{\Nr,\kappa}_\mu[g_\ell](\tau)=g_\ell(\tau)\,
\frac{D^\ell \ths{\Nr,\kappa}_{\mu}(\tau)}{\eta^{4\kappa \Nr^3 +\Nr c_2}(\tau)}\, ,
\label{exp1}
\ee
where $g_\ell(\tau)$ are modular forms of weight
\be
w_\ell=2\kappa \Nr^3+\hf\, \Nr c_2-2-2\ell-\delta^{(2)}_{\kappa r-1}\ .
\ee
Here, $D$ is the Serre derivative, acting on holomorphic modular forms of weight $w$
through  $D = \q \partial_{\q} - \frac{w}{12} E_2$, and $E_2$ is the normalized quasi-modular Eisenstein series.
The functions \eqref{exp1} satisfy all required properties and produce
the desired basis upon choosing an appropriate basis of modular forms $g_\ell$ of weight $w_\ell$.
In particular, since $w_\ell$ is an even integer\footnote{The reason of taking the derivative w.r.t. $z$ for odd $\kappa$
in the definition \eqref{defxiforms} was precisely to ensure this property.},
$g_\ell$ themselves can be represented as polynomials in Eisenstein series $E_4$ and $E_6$.
As a result, any $f\in \scM_\Nr(\CY)$ can be represented as
\be
f_\mu=\sum_{\ell=0}^{\mm} \(\sum_{k=0}^{k_\ell} c_{\ell,k}
\, E_4^{\lfloor w_\ell/4 \rfloor-\eps_\ell-3k}\, E_6^{2k+\eps_\ell}\)
\frac{D^\ell \ths{\Nr,\kappa}_{\mu}(\tau)}{\eta^{4\kappa \Nr^3 +\Nr c_2}(\tau)}\, ,
\label{decomp-modform}
\ee
where $k_\ell=\lfloor w_\ell/12 \rfloor-\delta^{(12)}_{w_\ell-2}$,
$\eps_\ell=\delta^{(2)}_{w_\ell/2-1}$ and $\mm$ is sufficiently large so that
$\sum_{\ell=0}^{\mm}(k_\ell+1)$ is not smaller than the number of polar terms.

\section{BPS indices for single D4-brane}
\label{sec-DT1}

As explained in \S\ref{sec-DT}, the functions $h_{1,\mu}$ are VV modular forms and therefore
they are fixed by their polar terms. In \S\ref{subsec-polar1}, we propose an Ansatz for these terms
and in \S\ref{subsec-result1} we present the results on the reconstruction of the generating functions $h_{1,\mu}$
on the basis of this Ansatz for 13 one-parameter CICY threefolds.

\subsection{Polar terms}
\label{subsec-polar1}

The BPS indices appearing in the polar terms of the generating functions $h_{\Nr,\mu}$
count black hole states with positive invariant $\hq_0$.
Since the area of a single-centered black hole horizon in $\cN=2$ supergravity is
given by $S=2\pi \sqrt{-\hq_0  p^3}$ with $p^3=\kappa_{abc}p^a p^b p^c>0$
\cite{Shmakova:1996nz,Maldacena:1997de}, such  single-centered solutions cannot contribute to polar terms.
Thus, only multi-centered bound states can contribute to such indices.\footnote{One may wonder then
why polar terms are non-vanishing given that there are no bound states at the attractor point
(except for the so-called scaling solutions which require at least three constituents).
However, the BPS indices entering the definition of the generating functions $h_{p,\mu}$ \eqref{defhDT}
are evaluated at the large volume attractor point, and will in general differ from the genuine attractor indices.}
In \cite{Denef:2007vg} it was shown that such  contributions arise from bound
states of D6 and anti-D6 branes with vanishing total D6-charge. Moreover, it was observed that
the `most polar terms`, i.e. the ones with $\hq_0$ sufficiently close to $\hq_0^{\rm max}$,
appear to receive contributions from a single D6-$\overline{\rm D6}$ pair {\it only}.
For the one-parameter  threefolds $X_5$, $X_6$ and $X_8$ and unit D4-brane charge,
this property was confirmed for all polar terms in \cite{Collinucci:2008ht,VanHerck:2009ww}.
These observations suggest the following
\begin{assume}
The polar coefficients in $h_{1,\mu}$ count the number of bound states of the form
\be
\begin{array}{cl}
\mbox{D6}\mbox{--}\mu\mbox{D2}\mbox{--}n\mbox{D0}
\ + \  \overline{\mbox{D6}[-1]}\quad &\mbox{for }0\le\mu\le \kappa/2
\\
\mbox{D6}[1] \ +\  \overline{\mbox{D6}\mbox{--}(-\mu)\mbox{D2}\mbox{--}(-n)\mbox{D0}} \quad &\mbox{for } 0\le -\mu\le \kappa/2,
\end{array}
\label{boundstates}
\ee
where ${\rm D6}[\Nr]$ denotes D6-brane with $\Nr$ units of D4-flux induced by spectral flow.
\label{conj1}
\end{assume}

Let us evaluate the degeneracy of these bound states explicitly.
For the sake of generality, and in order to discuss possible extensions in the following sections,
we will consider more general configurations of the form
\be
N\mbox{D6}[\Nr_1]\mbox{--}m_1\mbox{D2}\mbox{--}n_1\mbox{D0}
+\overline{N\mbox{D6}[-\Nr_2]\mbox{--}m_2\mbox{D2}\mbox{--}(-n_2)\mbox{D0}}.
\label{D6D6}
\ee
The contribution to the BPS index from a bound state with charges $\gamma_1+\gamma_2=\gamma$
is given by the primitive wall-crossing formula
\be
\Delta\Omega(\gamma) =(-1)^{\langle \gamma_1,\gamma_2\rangle} \langle \gamma_1,\gamma_2\rangle\,
\Omega(\gamma_1;z_{12})  \, \Omega(\gamma_2;z_{12}),
\label{primWC}
\ee
where $\langle \gamma_1,\gamma_2\rangle=q_{1,0}p_2^0+q_{1,a}p_2^a-(1\leftrightarrow 2)$ is the Dirac product of charges,
and the BPS indices on the r.h.s. are evaluated at the point $z_{12}$ in the moduli space where the attractor flow
corresponding to the charge $\gamma$ hits the wall of marginal stability corresponding to the decay of the bound state.
The charge vectors of the constituents in \eqref{D6D6} can be obtained by applying the spectral flow \eqref{specflowD6} to the charge vector
describing a $N$D6-$m$D2-$n$D0 bound state which is, consistently with the charge quantization \eqref{fractionalshiftsD5}, given by
\be
\gamma(N,m,n)=\(N,0,m-\frac{N}{24}\, c_2 ,-n\).
\label{chargeD6}
\ee
Then the spectral flow with $\epsilon=r$ gives
\be
\gamma(N,m,n)[\Nr]=\(N,N\Nr,m-\frac{N c_2}{24}
-\frac{N\kappa}{2}\, \Nr^2,-n-rm+\frac{N c_2}{24}\, \Nr+\frac{N\kappa}{6}\, \Nr^3\).
\label{chargeD6D4}
\ee
Choosing $\gamma_1=\gamma(N,m_1,n_1)[\Nr_1]$ and $\gamma_2=-\gamma(N,m_2,-n_2)[-\Nr_2]$, we obtain that the total charge reads
\be
\gamma
=\(0,N\Nr,m-\frac{N\kappa}{2}\, \Nr\bar \Nr ,-n-\hf\,(\bar \Nr m+\Nr\bar m)+\frac{N }{24}\,
\chi(\Nr\cD)+\frac{N\kappa }{8}\, \Nr\bar \Nr^2\),
\label{totalcharge}
\ee
where $\Nr=\Nr_1+\Nr_2$, $\bar \Nr=\Nr_1-\Nr_2$, $n=n_1+n_2$, $m=m_1-m_2$, $\bar m=m_1+m_2$,
and $\chi(\Nr\cD)$ is given in \eqref{defchiD-one}.
The invariant charge \eqref{defqhat} evaluates to
\be
\hq_0=\frac{N}{24}\, \chi(r\cD)-\frac{m^2}{2\kappa r N }-\frac{r\bar m}{2}-n,
\label{hqD4}
\ee
and the Dirac product of the charges of the two bound states is
\be
\langle\gamma_1, \gamma_2\rangle=-N^2\cI_\Nr+N(\Nr\bar m+n),
\label{DiracD6D6}
\ee
where $\cI_\Nr$ is given in \eqref{chiODN}.
Note that both $\hq_0$ and $\langle\gamma_1, \gamma_2\rangle$ do not depend on the parameter $\bar \Nr$.
This is consistent with the fact that under the spectral flow \eqref{specflowD6} acting on the charge vector \eqref{totalcharge},
this parameter is shifted by $2\epsilon$ so that one can always set it either to 0 or 1.
Substituting \eqref{DiracD6D6} into \eqref{primWC} gives the contribution to the BPS index.

According to our Assumption \ref{conj1}, we are interested in much simpler configurations where
$N=\Nr=1$, $m=\mu$, $\bar m=|\mu|$ (for both ranges of $\mu$), in which case
\be
\begin{split}
\bOm_{1,\mu}(\hq_0)=&\, (-1)^{n+|\mu|+\cI_1+1}\(\cI_1-|\mu|-n\)\Omega(\gamma_1;z_{12})\Omega(\gamma_2;z_{12}),
\\
\mbox{where} \qquad &
\hq_0=\frac{\chi(\cD)}{24}-\frac{\mu^2}{2\kappa }-\frac{|\mu|}{2}-n >0,
\end{split}
\ee
and the problem reduces to evaluating the BPS indices $\Omega(\gamma_i;z_{12})$, $i=1,2$.
We further assume
\begin{assume}
The BPS indices $\Omega(\gamma_i;z_{12})$ coincide with their values at large volume, i.e.
$\Omega(\gamma_i;z_{12})=\Omega(\gamma_i;z^a_\infty(\gamma_i))$.
\label{conj2}
\end{assume}
This conjecture implies that the BPS indices coincide with the standard rank 1 DT invariants $DT(|\mu|,n)$,
counting bound states of a single D6-brane with $|\mu|$ D2-branes and $\pm n$
D0-branes  (see \S\ref{sec_GVDT}). In the present case, either $\gamma_1$ or $\gamma_2$ corresponds
to a pure (anti-)D6-brane and the corresponding invariant $DT(0,0)=1$.
Thus, we arrive at the following expression for the polar part of $h_{1,\mu}$\footnote{Note that the second argument of $DT$ is
given by $n$ for both cases in \eqref{boundstates}. The reason for this is that $\sgn(b)=-\sgn(\mu)$ due to \eqref{lvolatt}
and therefore we must use the definitions \eqref{defDTa} and \eqref{defDTb} in the first and second cases, respectively.}
\be
\label{hpolar}
\hpol_{1,\mu} = \q^{-\frac{\chi(\cD)}{24}+\frac{\mu^2}{2\kappa}+\frac{|\mu|}{2}}
 \sum_{n\in \IZ\; :\; \hq_0>0} (-1)^{n+|\mu|+\cI_1+1}(\cI_1- |\mu|-n)\, DT(|\mu|,n)\, \q^n .
\ee
Several remarks about this formula are in order:
\begin{itemize}
\item Eq. \eqref{hpolar} is manifestly consistent with the symmetry $\mu\mapsto -\mu$, and
expected to hold in the range
$-\kappa/2\le \mu\le \kappa/2$.

\item
Note that the sum  is finite because $n$ is bounded from above by the condition $\hq_0>0$,
and from below due to the vanishing of $DT(|\mu|,n)$ for large negative $n$.
In fact, requiring that the most polar term arises for the component $\mu=0$
(in which case $n=0$) leads to a lower bound
\be
n \geq  - \frac{\mu^2}{2\kappa}- \frac{\mu}{2}
\ee
on the possible non-vanishing DT invariants $DT(\mu,n)$, which in turn implies an upper bound on the genus
\be
\label{Castel}
g \leq  \frac{Q^2}{2\kappa}+ \frac{Q}{2} + 1
\ee
for non-vanishing GV invariants $N_Q^{(g)}$. This Castelnuovo-type condition is well known
to hold for the quintic \cite{Katz:1999xq,Huang:2007sb}, and is consistent with the tables
of GV invariants in \cite{Huang:2006hq}. We conjecture that \eqref{Castel} is in fact valid
for all for one-parameter CICYs.

\item
As we discuss in Appendix  \ref{sec_math}, the leading polar coefficient in \eqref{hpolar}
(arising from $\mu=n=0$) agrees with results in the mathematical literature
\cite{Toda:2011aa,Feyzbakhsh:2022ydn}, and subleading polar coefficients are also in broad
agreement.
\end{itemize}
In the following, we shall take the formula  \eqref{hpolar} as our Ansatz for
the polar terms that we use to reconstruct the generating functions $h_{1,\mu}$.
A tentative generalization to higher rank is discussed in \S\ref{subsec-polar2}.

\subsection{Results}
\label{subsec-result1}

We perform the reconstruction of $h_{1,\mu}$ from their polar part for 13 CICY threefolds given in Table \ref{table1}.
To this end, for each of these threefolds, we construct the linear combinations \eqref{decomp-modform}
(for an appropriately chosen $\mm$) and match their polar terms against the ones predicted by \eqref{hpolar}
where DT invariants $DT(|\mu|,n)$ are calculated from the known sets of GV invariants in \S\ref{sec-genfun1}.
This provides a system of linear equations on the coefficients $c_{m,k}$.
If this system has a solution, it gives rise to a VV modular form with the desired polar part.
We further require that
the coefficients of its Fourier expansion should be integer, in order to be interpretable
as BPS indices (or rank-zero DT invariants).\footnote{Note that for $\Nr=1$,
the D4-D2-D0 charge is always primitive and therefore the rational BPS indices
appearing in \eqref{defhDT} coincide with the integer valued ones.}
Here are the results of our analysis:
\begin{itemize}
\item
For 10 out of 13 threefolds, the system of equations on $c_{m,k}$ turns out to have a unique solution with integer Fourier coefficients.
The explicit expressions for the resulting VV modular forms and the first terms in their Fourier expansion are given in \S\ref{sec-genfun1}.
For $X_5$, $X_6$, $X_8$, $X_{10}$ and $X_{3,3}$, our results reproduce those in
\cite{Gaiotto:2006wm,Gaiotto:2007cd,Collinucci:2008ht,VanHerck:2009ww}.

\item
For the remaining 3 models $X_{4,2}$, $X_{3,2,2}$ and $X_{2,2,2,2}$, the polar coefficients
do not allow for the existence of a VV modular form, which indicates that
our Ansatz for the polar terms \eqref{hpolar} needs to be modified in these cases.

\item
In those cases, it is easy to tweak  the polar terms in an  {\it ad hoc} fashion so as to allow for a
solution with integer coefficients.
In particular, this can be done in a `minimal' way by changing only the polar terms for the maximal D2-brane charge $|\mu|=\kappa/2$
and, in the case of $X_{2,2,2,2}$, also for $|\mu|\ge\kappa/2-2$.

\end{itemize}

In view of this last point, one might be reluctant to trust
the solutions found in the `10 cases that work', especially since in most of them the polar terms
do not need to satisfy any constraints to generate a modular form.
However, there are three observations in favor of our results:
\begin{itemize}
\item
As indicated above, they reproduce known results for models already studied in the literature.

\item
For $X_{3,3}$ and $X_{4,4}$, there are in fact modular constraints on polar coefficients,
which turn out to be satisfied by our Ansatz, thanks to uncanny relations between GV
invariants (see \eqref{uncanny1} and \eqref{uncanny2}).

\item
All solutions satisfy the condition of having integer Fourier coefficients, which was not guaranteed at all.

\end{itemize}
This provides some evidence that our generating series may be correct.
However, of course, it leaves open the question as to why and how our ansatz should
be modified in the remaining three cases and, in particular, whether the minimal modification that we propose is indeed correct.

\section{BPS indices for D4-brane charge 2}
\label{sec-DT2}

In this section we go beyond the rank one case and explain how to fix the generating functions $h_{2,\mu}$,
assuming that $h_{1,\mu}$ have been previously determined.
In \S\ref{subsec-strategy} we present our general strategy and in \S\ref{subsec-explicit} we provide an explicit algorithm.
Unfortunately, our lack of control on the polar terms does not allow us to implement this algorithm successfully
in concrete examples.

\subsection{General strategy}
\label{subsec-strategy}

The generating functions $h_{2,\mu}$ are VV mixed mock modular forms with completion $\whh_{2,\mu}$ given by \eqref{whh2}.
Such functions are not uniquely specified by the polar part,
unless one also specifies the shadow determining the modular anomaly, or equivalently the holomorphic anomaly of its completion.
The latter being an inhomogeneous linear equation,
its general solution is a sum of a particular solution and a solution of the corresponding homogeneous equation.
In our case the homogeneous solution is nothing but a genuine VV modular form.

This observation suggests the following method to reconstruct the generating function from its polar part.
First, we need to produce a function $\han_{2,\mu}$ that has the same modular anomaly as $h_{2,\mu}$, which will play
the role of the particular solution for the modular anomaly equation. Then the full generating function
is a sum of $\han_{2,\mu}$ and a VV modular form $\hh_{2,\mu}$,
\be
h_{2,\mu}=\han_{2,\mu}+\hh_{2,\mu}.
\label{hhh}
\ee
At the second step, this unknown modular form can be determined by its polar terms which are obtained as a difference
of the polar terms of $h_{2,\mu}$ (to be determined independently) and the polar terms of $\han_{2,\mu}$
(which can be read off from its explicit expression).

Thus, leaving aside the issue of fixing the polar part of $h_{2,\mu}$, which we do not attempt
in this paper,
the problem reduces to finding a function $\han_{2,\mu}$ with the anomaly determined by the shadow of $h_{2,\mu}$,
which in turn can be derived from the holomorphic anomaly of its completion $\whh_{2,\mu}$.
Since the anomalous term in \eqref{whh2} has a factorized form, it is natural
to look for $\han_{2,\mu}$ of the same form, namely
\be
\han_{2,\mu}=\sum_{\mu_1,\mu_2=0}^{\kappa-1}g_{2,\mu,\mu_1,\mu_2}\, h_{1,\mu_1}\,h_{1,\mu_2},
\label{defnorm}
\ee
where $h_{1,\mu}$ are the generating functions considered in the previous section.
The `normalized functions' $g_{2,\mu,\mu_1,\mu_2}$ should be chosen such that their completions defined by
\be
\whg_{2,\mu,\mu_1,\mu_2}=g_{2,\mu,\mu_1,\mu_2}+R_{\mu,\mu_1\mu_2},
\ee
where $R_{\mu,\mu_1\mu_2}$ is the same function \eqref{defR} that appears in the expression for $\whh_{2,\mu}$,
must transform as VV modular forms of weight 3/2 and the following multiplier system
\be
\begin{split}
M_{\mu,\mu_1,\mu_2;\nu,\nu_1,\nu_2}(T)
=&\,
e^{\pi\I\( \frac{1}{\kappa}(\hf\, \mu^2-\mu_1^2-\mu_2^2)-\mu_1-\mu_2 -\frac{\kappa}{2}\)}
\,\delta_{\mu\nu}\,\delta_{\mu_1\nu_1}\,\delta_{\mu_2\nu_2},
\\
M_{\mu,\mu_1,\mu_2;\nu,\nu_1,\nu_2}(S)
=&\, \frac{\sqrt{\I}(-1)^\kappa}{\sqrt{2}\kappa^{3/2}}\,
e^{\frac{2\pi\I}{\kappa}\(\mu_1\nu_1+\mu_2 \nu_2- \hf\,\mu \nu\)}.
\end{split}
\label{STref}
\ee

Furthermore, the function $R_{\mu,\mu_1\mu_2}$ encoding the anomaly is also of the special form \eqref{defR}
so that it is expressed through a vector like object.
The fact that $g_{2,\mu,\mu_1,\mu_2}$ can be taken in the same form is established by the following proposition:
\begin{proposition}
If $\Gi{\kappa}_\mu$ ($\mu=0,\dots,2\kappa-1$) transforms with the multiplier system
\be
\begin{split}
\Mi{\kappa}_{\mu\nu}(T)
=&\,
e^{-\frac{\pi\I}{2\kappa}\, \mu^2}\,\delta_{\mu\nu},
\\
\Mi{\kappa}_{\mu\nu}(S)
=&\, \frac{\sqrt{\I}}{\sqrt{2\kappa}}\,
e^{\frac{\pi\I}{\kappa}\,\mu \nu},
\end{split}
\label{STusual}
\ee
then
\be
g_{2,\mu,\mu_1,\mu_2}=\delta^{(\kappa)}_{\mu_1+\mu_2-\mu}(-1)^{\mu'} \Gi{\kappa}_{\mu'}
\label{ansatz-g}
\ee
transforms with the multiplier system \eqref{STref}.
\end{proposition}
\begin{proof}
Let us verify the T-transformation. Taking into account the $\delta$-symbol in \eqref{ansatz-g},
the phase factor required to be produced by this transformation from \eqref{STref} is found to be (here $\lambda\in \IZ$)
\be
\begin{split}
&\,
e^{\pi\I\( \frac{1}{\kappa}(\hf\, \mu^2-\mu_1^2-\mu_2^2)-\mu_1-\mu_2 -\frac{\kappa}{2}\)}\stackrel{\mu_1+\mu_2=\mu+\lambda\kappa}{=}
e^{\pi\I\( \frac{1}{\kappa}\(\hf\, \mu^2-\mu_1^2-(\mu-\mu_1+\lambda\kappa)^2\)-\mu -\lambda\kappa-\frac{\kappa}{2}\)}
\\
=&\,
e^{-\pi\I \kappa(\lambda+\lambda^2)}\,
e^{-\frac{\pi\I}{2\kappa}\,(\mu-2\mu_1+\kappa)^2}=e^{-\frac{\pi\I}{2\kappa}\,\mu'^2},
\end{split}
\ee
which indeed coincides with the phase factor in \eqref{STusual}.

Similarly, by computing the Fourier transform implied by \eqref{STref}, one finds
\bea
&&
\frac{\sqrt{\I}(-1)^\kappa}{\kappa\sqrt{2\kappa}}\sum_{\nu,\nu_1,\nu_2} e^{\frac{\pi\I}{\kappa}\(-\mu\nu+2\mu_1\nu_1+2\mu_2\nu_2\)}
\,\delta^{(\kappa)}_{\nu_1+\nu_2-\nu}\,(-1)^{ \nu'}\, \Gi{\kappa}_{\nu'}
\nn\\
&=& \frac{\sqrt{\I}(-1)^\kappa}{\kappa\sqrt{2\kappa}}\,\sum_{\nu,\nu_1}e^{\frac{\pi\I}{\kappa}\((2\mu_2-\mu)\nu+2(\mu_1-\mu_2)\nu_1\)}
\,(-1)^{ \nu'}\, \Gi{\kappa}_{\nu'}
\nn\\
&=& \frac{\sqrt{\I}}{\kappa\sqrt{2\kappa}}\,\sum_{\nu',\nu_1}
e^{\frac{\pi\I}{\kappa}(2\mu_2-\mu)(\nu'-\kappa)+\frac{2\pi\I}{\kappa}(\mu_1+\mu_2-\mu)\nu_1}
\,(-1)^{ \nu'-\kappa}\,\Gi{\kappa}_{\nu'}
\nn\\
&=&\delta^{(\kappa)}_{\mu_1+\mu_2-\mu}\,(-1)^{\mu'}
\, \frac{\sqrt{\I}}{\sqrt{2\kappa}}\sum_{\nu'}e^{\frac{\pi\I}{\kappa}\,\mu'\nu'}\, \Gi{\kappa}_{\nu'}.
\label{Fthetamu}
\eea
This result is perfectly consistent with the S-duality transformation of $\Gi{\kappa}_\mu$ implied by \eqref{STusual}.
\end{proof}

Due to this proposition, choosing the functions $g_{2,\mu,\mu_1,\mu_2}$ in the form \eqref{ansatz-g},
we finally reduce the problem to finding a VV mock modular form $\Gi{\kappa}_\mu$ such that its completion,
transforming with weight $3/2$ and multiplier system \eqref{STusual}, is given by
\be
\whGi{\kappa}_\mu=\Gi{\kappa}_\mu+\Thi{\kappa}_\mu,
\label{relGTh}
\ee
where $\Thi{\kappa}_\mu$ is defined in \eqref{defRmu}.
The original generating function is then obtained by substituting \eqref{defnorm} and \eqref{ansatz-g} into \eqref{hhh}
leading to
\be
\begin{split}
h_{2,\mu}
=&\, \hh_{2,\mu}+\sum_{\mu_1=0}^{\kappa-1}
(-1)^{\mu-2\mu_1+\kappa} \Gi{\kappa}_{\mu-2\mu_1+\kappa}\, h_{1,\mu_1}\,h_{1,\mu-\mu_1}\,.
\end{split}
\label{defnorm-G}
\ee
The holomorphic ambiguity $\hh_{2,\mu}$ can be fixed by matching polar terms.

\subsection{Explicit construction}
\label{subsec-explicit}

The upshot  of the previous subsection is that we reduced the problem of finding a VV {\it mixed} mock modular form,
with a modular anomaly depending on the generating functions for $\Nr=1$,
to a similar problem for the usual VV mock modular form $\Gi{\kappa}_\mu$ with an anomaly specified by $\Thi{\kappa}_\mu$ \eqref{defRmu}.
This is a much simpler problem which we can actually solve using known results in the literature.

The key observation is that the shadow \eqref{holanom-Th} of the completion $\whGi{\kappa}_\mu$ is,
up to a trivial factor $\tau_2^{-3/2}$, the complex conjugate of a simple unary theta series.
Furthermore, for $\kappa=1$, it is identical (up to a factor of 3) with the shadow
of the generating series  of Hurwitz class numbers, which (not coincidentally) appears
in the context of rank 2 Vafa-Witten invariants
on $\IP^2$  \cite[Eq.(4.32)]{Vafa:1994tf}\footnote{The connection between Hurwitz class numbers
and moduli spaces of rank 2 semi-stable sheaves on
$\IP^2$ was derived earlier in the mathematics literature \cite{Klyachko:1991,Yoshioka:1994},
and the mock modular properties of the corresponding generating series were established
in \cite{Zagier:1975,hirzebruch1976intersection}.}.
Thus, for $\kappa=1$ we can simply choose
\be
\Gi{1}_\mu=H_\mu,
\label{solG1}
\ee
where $H_\mu$ is the standard (doublet of) generating series of Hurwitz class numbers,
which starts with the following coefficients
\be
\label{H01}
\begin{split}
H_0(\tau) =&\, -\frac{1}{12} + \frac{\q}{2}+\q^2+\frac{4 \q^3}{3}+\frac{3 \q^4}{2}+2 \q^5+2 \q^6+2 \q^7+3 \q^8+\frac{5
\q^9}{2}+2 \q^{10}
+\dots\, ,
\\
H_1(\tau) =&\, \q^{\frac34} \left( \frac{1}{3}+\q+\q^2+2
\q^3+\q^4+3 \q^5+\frac{4 \q^6}{3}+3 \q^7+2 \q^8+4 \q^9+\q^{10}
 + \dots \right).
\end{split}
\ee

In order to upgrade this solution to $\kappa>1$, we need an operator acting on VV modular forms which
i)  preserves their weight but increases the dimension of the vector space in which they are valued,
in particular, mapping the multiplier system $\Mi{1}$ to $\Mi{\kappa}$ \eqref{STusual}, and
ii) maps $\Thi{1}$ to $\Thi{\kappa}$.
In Appendix \S\ref{sec-Hecke}, we show that when $\kappa$ is a prime number, these properties are satisfied by
a generalized Hecke operator $\cT_\kappa$ introduced in \cite{Bouchard:2016lfg,Bouchard:2018pem}.
Unfortunately, for $\kappa$ non-prime
it fails to satisfy the second property that ensures that the anomalies are properly matched.
Nevertheless, when $\kappa$ is a power of a prime number\footnote{Out of the list of 13 CICY,
this rules out $X_{4,3}$ and $X_{3,2,2}$, for which the construction of $\cT'_\kappa$ is left as an open problem.},
it is possible to cure the problem and modify $\cT_\kappa$ into an operator
$\cT'_\kappa$ such that
\be
\Gi{\kappa}_\mu=(\cT'_\kappa[H])_\mu.
\label{GkapH}
\ee
Substituting \eqref{GkapH} into \eqref{defnorm-G},
we finally arrive at the following representation for the generating functions
\be
h_{2,\mu}= \hh_{2,\mu}+ \sum_{\mu_1=0}^{\kappa-1}
(-1)^{\mu-2\mu_1+\kappa} (\cT'_\kappa[H])_{\mu-2\mu_1+\kappa}\, h_{1,\mu_1}\,h_{1,\mu-\mu_1},
\label{genfun-Hecke-p}
\ee
where the action of $\cT'_\kappa$ is defined by \eqref{defHecke-exp} and \eqref{defTp}. Thus,
we only need to fix the holomorphic modular ambiguity $\hh_{2,\mu}$, which can be determined
from its polar part.

Let us assume that we know the  polar part of the generating series of  {\it integer} DT-invariants
\be
\hint_{2,\mu}=\sum_{\hq_0\ge 0}\Omega_{2,\mu}(\hq_0)\, \q^{-\hq_0},
\label{polarh2}
\ee
namely all integer coefficients $\Omega_{p,\mu}(\hq_0)$ for $p=2$ and $\hq_0>0$.
The generating series \eqref{polarh2} differs from the generating function
$h_{2,\mu}$ of rational DT-invariants \eqref{defntilde} due to
the contribution of non-primitive charges representable as $\gamma=2\gamma'$.
Since the general form of the charges with $p=2$ and $p=1$ is
\be
\gamma=(0,2,2\kappa\eps +\mu+2\kappa,q_0),
\qquad
\gamma'=\Bigl(0,1,\kappa\eps+\mu'+\frac{\kappa}{2}\,,q'_0\Bigr),
\ee
one must have
\be
\mu'=\hf\,(\mu+\kappa) \in\IZ\, ,
\qquad q_0'=\hf\, q_0\in\IZ\, .
\ee
Therefore, the relation between the generating functions of rational and integer BPS indices reads
\be
h_{2,\mu}(\tau)=\hint_{2,\mu}(\tau)+\frac14\, \delta^{(2)}_{\mu+\kappa}\, h_{1,\frac{\mu+\kappa}{2}}(2\tau).
\label{polparth}
\ee
Substituting this relation into \eqref{genfun-Hecke-p}, we find that the polar part of the holomorphic ambiguity
$\hh_{2,\mu}$ is given by the polar part of
\be
\begin{split}
\hint_{2,\mu}(\tau)- \sum_{\mu_1=0}^{\kappa-1}
(-1)^{\mu-2\mu_1+\kappa} (\cT'_\kappa[H])_{\mu-2\mu_1+\kappa}(\tau)\, h_{1,\mu_1}(\tau)\,h_{1,\mu-\mu_1}(\tau)
+ \frac14\, \delta^{(2)}_{\mu+\kappa}\, h_{1,\frac{\mu+\kappa}{2}}(2\tau),
\end{split}
\ee
which is entirely determined by the polar part of $\hint_{2,\mu}$ (which serves as input)
and by $h_{1,\mu}$ (which by assumption has been previously determined). Assuming that a VV modular form  $\hh_{2,\mu}$
with the required polar part exists, we can plug it into \eqref{genfun-Hecke-p} to obtain
the generating functions $h_{2,\mu}$ of rational DT invariants, and
finally obtain the generating functions $\hint_{2,\mu}$ of integer DT invariants
 via \eqref{polparth}.
If no such VV modular form $\hh_{2,\mu}$ exists, or if  the Fourier coefficients of $\hint_{2,\mu}$
turn out to not be integer,
one must conclude that the proposed polar part is incorrect,
or that a mistake has been made in the previous step of determining $h_{1,\mu}$.

\subsection{A naive attempt \label{subsec-polar2}}

Given our partial success at rank 1, it is natural to extend the Ansatz \eqref{hpolar}
to higher D4-brane charge, by keeping only contributions from a single \DDb pair
(i.e. $N=1$) with $\Nr=\Nr_1+\Nr_2>1$ units of flux.
Then the same reasoning as in \S\ref{subsec-polar1} leads to the proposal
\be
\label{hpolargenr}
\hpol_{\Nr,\mu} \stackrel{?}{=} \q^{-\frac{\chi(\Nr\cD)}{24}+\frac{\mu^2}{2\Nr\kappa}+\frac{\Nr\mu}{2}}
\sum_{n\in \IZ\; :\; \hq_0>0} (-1)^{n+\Nr\mu+\cI_\Nr+1}(\cI_\Nr-\Nr \mu-n)\, DT(\mu,n)\, \q^n.
\ee
Unfortunately, setting $\Nr=2$, restricting to the 9 models for which the rank 1 invariants
had been determined and $\kappa$ is a power of a prime number,
and applying the algorithm outlined in the previous subsection, we find that no solution
$\hh_{2,\mu}$ with the required polar terms exists whenever the polar part is constrained
(i.e. $C_2>0$ in Table \ref{table1}), or that the solution does not lead to integer coefficients in $\hint_{2,\mu}$.
This suggests that the Ansatz \eqref{hpolargenr} misses some contributions, as we discuss in the next Section.

\section{Discussion}
\label{sec-disc}

In this paper we used modular properties of the generating series of D4-D2-D0 BPS indices
to determine these functions explicitly in the case of compact CY threefolds with $b_2(\CY)=1$.
In this case, the generating functions depend on one positive integer $r$ ---
the wrapping number of D4-brane along the primitive divisor,  or D4-brane charge for short.
For $\Nr=1$, when the generating functions are VV modular forms, we proposed an Ansatz \eqref{hpolar} for their polar terms,
which generalizes the known results in the literature \cite{Gaiotto:2006wm,Gaiotto:2007cd,Collinucci:2008ht,VanHerck:2009ww}.
It allowed us to produce the generating functions $h_{1,\mu}$ for 10 out of the list of
13 one-parameter CICY threefolds.

For $\Nr=2$, when the generating functions are VV mixed mock modular forms, we constructed an
explicit solution to the corresponding modular anomaly equation
by applying a suitable Hecke operator on the generating function of Hurwitz class numbers,
which arises in the similar problem of rank 2 Vafa-Witten invariants on $\IP^2$.
This determines $h_{2,\mu}$ up to a holomorphic VV modular form which is supposed to be fixed by the polar terms.
In principle, the same strategy would also work for $\Nr>2$, using the rank $\Nr$
VW invariants on  $\IP^2$ determined in \cite{Manschot:2010nc,Manschot:2011ym,Manschot:2014cca}
(see also \cite{Alexandrov:2020dyy}) as a starting point, although the construction
of a solution to the modular anomaly equation is likely to be more complicated.
However, already for $\Nr=2$, we found that the naive extension \eqref{hpolargenr}
of the Ansatz \eqref{hpolar} does not work. The determination of the correct polar terms (both for rank 1
and higher) is therefore the main open problem for future investigations.

Without trying to solve this problem here, let us discuss the possible origin
of the contributions that are missed by the naive Ansatz \eqref{hpolargenr}.
Firstly, for $\Nr>1$ it is natural to expect that contributions from
$N$ D6-$\overline{\rm D6}$ pairs with $1\le N \le \Nr$ may become relevant.
Some of these contributions can be easily deduced from the computation presented in \S\ref{subsec-polar1}
by combining equations \eqref{primWC}, \eqref{hqD4} and \eqref{DiracD6D6}. As in the $N=1$ case,
the BPS indices for each of the two constituents can then be related to rank $N$ Donaldson-Thomas invariants.
Those are in principle determined by  rank 1 DT invariants \cite{Feyzbakhsh:2021nds,Feyzbakhsh:2022ydn},
although it may be difficult to determine them in practice. It is also
possible that more complicated bound states need to be taken into account where D4-brane charge
is not generated by the spectral flow as in \eqref{chargeD6D4}, but at least partially produced by D4-flux on a D6-brane.
Then the BPS indices of the constituents are given by generalized DT invariants which are rarely known explicitly.
An even more complicated scenario would involve contributions from bound states with multiple constituents,
for example, one with two  units of D6-branes and two with a single $\overline{\rm D6}$-brane.
In that case, it would be difficult to produce any general Ansatz and we would have to rely on a case by case analysis.

Second, despite some success, our Ansatz for $\Nr=1$ also needs confirmation and improvement,
as there are three CICY threefolds  for which it fails to produce a modular form.
This sheds doubt on its validity in other cases where it does produce a plausible result but is weakly constrained by modularity.
We would like to put forward a few observations pointing to possible resolutions:
\begin{itemize}
\item
The three offending cases correspond to one-parameter families with a singularity at $\psi=\infty$
of type C or M in the terminology of \cite{Joshi:2019nzi}, corresponding to a conifold singularity
at finite distance (in addition to the conifold singularity at $\psi=1$, which is common to all models),
or a maximal unipotent monodromy at infinite distance  (in addition to the large volume point at $\psi=0$,
common to all models). It is conceivable that such singularities give rise to new constituents
analogous to the \DDb bound states which could contribute to polar terms. In this respect, it is worth noting that
the Ansatz \eqref{hpolar} seems to work for
the models $X_{k,k}$ with $k=3,4,6$ having a $K$-type singularity at infinite distance.

\item
In Assumption \ref{conj1}, we assumed that D2 and D0-branes can only bind to the  D6 or
$\overline{\rm D6}$-brane, depending on the sign of $\mu$. As reviewed in Appendix \ref{sec_math},
the mathematical results of \cite{Toda:2011aa,Feyzbakhsh:2022ydn} indicate that this is not true in general,
and D2 and D0-branes may bind to both the D6 and $\overline{\rm D6}$-brane,
leading to terms quadratic in DT-invariants.

\item
Even in cases where Assumption \ref{conj1} is valid, Assumption \ref{conj2} may fail,
in the sense that the BPS indices of
the constituents might differ from the rank 1 DT invariants due to wall-crossing between
the large volume point and the point on the wall of marginal stability
at which they are to be evaluated. This is corroborated by
the fact that both Donaldson-Thomas and Pandharipande-Thomas invariants
enter in the mathematical results of \cite{Toda:2011aa,Feyzbakhsh:2022ydn}.

\item
As discussed in  \S\ref{sec-genfun1}, one may  modify the polar coefficients in an
{\it ad hoc} way so as to  produce a modular form with integer coefficients
for the three CYs where the original Ansatz fails.
An intriguing observation is that it suffices to modify only those coefficients that
correspond to {\it negative} D0-brane charge $n<0$.
Furthermore, the only other case where non-vanishing polar coefficients with $n<0$ arise
is the leading polar term in $h_{1,2}$ for $X_{6,2}$ (see \eqref{exp-hX62}),
but that coefficient can be changed without affecting modularity.
Thus, it might be that  contributions from bound states with $n<0$ need to be treated differently.
If so, this would also explain why our Ansatz works in all other cases where negative $n$ does not appear.

\end{itemize}

Finally, it might happen that the contributions from multiple D6-$\overline{\rm D6}$
pairs discussed above are also relevant for $\Nr=1$.
In that case, a careful analysis of multi-centered configurations of D6 and $\overline{\rm D6}$
(potentially including scaling solutions) will be needed,
and the modular generating series recorded in Appendix \ref{sec-genfun1} cannot be trusted.
We hope to return to the analysis of the polar coefficients in future work.

\section*{Acknowledgements}

We are grateful to  Frederik Denef, Charles Doran, Soheyla Feyzbakhsh and Albrecht Klemm for useful correspondence,
and to Aradhita Chattopadhyaya and Thorsten Schimannek for discussions.
SA, NG and BP appreciate the hospitality of the Hamilton
Institute at Trinity College Dublin for hospitality during the course of this work.
This collaboration was supported by a Ulysses Award from the Irish Research Council
and French Ministry of Europe and Foreign Affairs. The research of NG is supported
by the Delta-Institute for Theoretical Physics (D-ITP) that is funded by the Dutch Ministry of Education,
Culture and Science (OCW). The research of JM is
supported by Laureate Award 15175 ``Modularity in Quantum Field Theory
and Gravity'' of the Irish Research Council. The research of BP is supported by Agence Nationale de la Recherche
under contract number ANR-21-CE31-0021.

\appendix

\section{Polar terms of vector valued modular forms}
\label{sec-constraints}

In this section, we determine the dimension of the space $\scM_\Nr(\CY)$
of weakly holomorphic VV modular forms  with Fourier expansion of the form
\be
\hh_{r,\mu} = \sum_{n\geq 0}  c_{\mu}(n) \,\q^{n-\Delta_\mu}
\ee
with exponents $\Delta_\mu$ specified in \eqref{Deltamu},
 transforming with weight $-3/2$ and multiplier system \eqref{multsys-h2} under
 $\tau\mapsto \frac{a\tau+b}{c\tau+d}$. Equivalently, we determine the number of linear constraints that the polar
 coefficients $c_\mu(n)$ with $n<\Delta_\mu$ must satisfy, in order to correspond to an element
$\hh_{r,\mu}\in \scM_\Nr(\CY)$.
For $r=1$, $\hh_{1,\mu}$ coincides with the generating series of DT invariants $h_{1,\mu}$,
whereas for $r>1$, it corresponds to the holomorphic ambiguity, as explained in \S\ref{sec-DT2}.
Throughout this section, we set $m=\kappa\,\Nr$, providing the normalization of the
quadratic form on the relevant lattice.

\subsection{Number of polar terms}

Since $\hh_{r,\mu}$ is by assumption invariant under $\mu\mapsto\mu+m$ and $\mu\mapsto -\mu$,
it consists of $d=\lceil \frac{m+1}{2} \rceil$ independent components. The number of polar terms is
therefore given by
\be
n_r(\CY)= \sum_{\mu=0}^{d-1 } \lceil \Delta_\mu \rceil.
\ee
It will be useful to rewrite this formula  as
\be
n_r(\CY)= -\frac{1}{2}\, \cI(M)+\sum_{\mu=0}^{d-1 }  \(\Delta_\mu+\frac{1}{2}-((\Delta_\mu))\),
\label{numpolterms}
\ee
where $\cI(M)$ is the number of exponents $\Delta_\mu$ which are integer, and
$((\,\cdot\,))$ is defined by
\begin{equation}
((x))=x-\frac{\lceil x \rceil + \lfloor x
  \rfloor}{2}=\left\{\begin{array}{ll} \xi-\frac{1}{2},\qquad  & {\rm
      if}\,\, x=\xi+\mathbb{Z},\,\, 0<\xi<1,\\ 0, & {\rm
      if}\,\, x\in \mathbb{Z}. \end{array} \right.
\end{equation}

\subsection{Number of constraints}

The constraints on polar terms of a weakly holomorphic modular form of weight $w$
originate from holomorphic cusp forms of dual weight $2-w$~\cite{Skoruppa85,Bantay:2007zz,Manschot:2008zb}.
To obtain the dimension of this space, hence the number of constraints,
one uses Selberg's trace formula, which gives the
difference of the dimension of the space $\cM_{w}(M)$ of VV holomorphic modular forms of weight $w$ and multiplier system
$M_{\mu\nu}$ and the dimension of the space $\cS_{2-w}(\bM)$ of cusp forms with complex conjugate multiplier system.
The trace formula gives \cite{Skoruppa85}
\begin{equation}
\label{dimCuspM}
{\rm dim}\!\left[\cS_{2-w}(\bM) \right]-{\rm dim}\!\[\cM_{w}(M)\]=A_{\rm s} + A_{\rm e}+A_{\rm p},
\end{equation}
where $A_{\rm s}$, $A_{\rm e}$ and $A_{\rm p}$ are the scalar, elliptic and parabolic contributions which are given by
\begin{equation}
\label{defAsep}
\begin{split}
A_{\rm s}=&\, \frac{1-w}{12}\,\chi_M(1),
\\
A_{\rm e}=&\, -\frac{1}{4}\,{\rm Re}\!\left[ e^{\frac{\pi\I w}{2}} \,\chi_M (S)\right]
+\frac{2}{3\sqrt{3}}\,{\rm Re}\!\left[ e^{\frac{\pi\I}{6}\,(2w-5)}\,\chi_M (ST)\right],
\\
A_{\rm p}=&\, -\frac{1}{2}\, \cI(M)-\sum_{\mu=0}^{d-1} ((\Delta_\mu)).
\end{split}
\end{equation}
Here, $\chi_M(g)$ denotes the character $\Tr M(g)$  of the action of $g\in SL(2,\IZ)$
on the $d$-dimensional vector space of components. Since the relevant weight $w=-3/2$ is negative,
the space $\cM_{-3/2}(M)$ of holomorphic VV modular forms is empty, and the r.h.s. of \eqref{dimCuspM}
gives directly the number of constraints on polar terms. Since
$\chi_M(1)=d$, it remains only to evaluate the elliptic contribution $A_{\rm e}$.

To this end, we introduce the Gauss sums $G(n,m)$
\be
G(n,m)=\sum_{\nu=1}^{m} e^{\frac{2\pi i n \nu^2}{m}}.
\ee
Recasting the $m\times m$ matrix $M_{\mu\nu}(S)$  in  \eqref{multsys-h2} into a
$d\times d$ matrix, we obtain
\be
\chi_M(S)=
\frac{(-1)^{\cI_\Nr}}{\sqrt{\I m}}
\begin{cases}
1-e^{-\frac{\pi \I m}{2}}
+\sum_{\nu=1}^{m/2} \(e^{\frac{2\pi \I  \nu^2}{m}}+e^{-\frac{2\pi \I  \nu^2}{m}}\), & m\ {\rm even},
\\
1+\sum_{\nu=1}^{(m-1)/2}\( e^{\frac{2\pi \I  \nu^2}{m}}+e^{-\frac{2\pi \I  \nu^2}{m}}\), & m \ {\rm odd},
\end{cases}
\ee
with $\cI_\Nr$ defined in (\ref{chiODN}). In either case, we find
\be
\chi_M(S)=
\frac{(-1)^{\cI_\Nr}}{2\sqrt{\I m}}\,\Bigl( G(1,m)+G(1,m)^*\Bigr).
\ee
Using the well-known values for the Gauss sum $G(1,m)$ \cite{Apostol1976},
\be
\label{G1m}
\begin{split}
G(1,m)=\left\{ \begin{array}{ll}
(1+\I )\sqrt{m}, & m=0\!\!\mod 4,
\\
\sqrt{m}, & m=1\!\!\mod 4,
\\
0, & m=2\!\!\mod 4,
\\
\I \sqrt{m}, & m=3\!\!\mod 4,
\end{array}\right.
\end{split}
\ee
we arrive at
\begin{equation}
\chi_M(S)=\left\{ \begin{array}{ll}
(-1)^{\cI_\Nr} e^{-\frac{\pi\I}{4}},\quad & m=0,1\!\!\mod 4,
\\
0, & m=2,3\!\!\mod 4.
\end{array} \right.
\label{chiS}
\end{equation}

Turning to $\chi_M(ST)$, we have
\be
\chi_M(ST)=\frac{(-1)^{\cI_\Nr}}{\sqrt{\I m}}\, e^{\frac{\pi\I}{4}\,mr^2 +\frac{\pi\I}{12}\, rc_2}
\begin{cases}
1-(-1)^{\frac{m\Nr}{2}} e^{-\frac{\pi \I  m}{4}}
+\sum_{\nu=1}^{m/2} (-1)^{\Nr\nu}\( e^{-\frac{\pi \I  \nu^2}{m}}
+e^{\frac{3\pi \I  \nu^2}{m}}\),  & m \ {\rm even},
\\
1 +\sum_{\nu=1}^{(m-1)/2} (-1)^{\Nr\nu }\( e^{-\frac{\pi \I  \nu^2}{m}}+e^{\frac{3\pi \I \nu^2}{m}} \),
& m \ {\rm odd}.
\end{cases}
\ee
For $m$ even, we can further simplify this to
\begin{equation}
\label{chiSTeven0}
\begin{split}
\chi_M(ST)&=\frac{e^{\frac{\pi\I}{3}\,mr^2}}{2\sqrt{\I m}}\sum_{\nu=1}^{m}
\(e^{-\frac{\pi \I  \nu^2}{m}}+e^{\frac{3\pi \I  \nu^2}{m}}\)
=\frac{e^{\frac{\pi\I}{3}\,mr^2}}{4\sqrt{\I m}}
\sum_{\nu=1}^{2m}\( e^{-\frac{2\pi \I\nu^2}{2m}}+e^{\frac{6\pi \I\nu^2}{2m}}\)
\\
&=\frac{e^{\frac{\pi\I}{3}\,mr^2}}{4\sqrt{\I m}}\, \Bigl(G(1,2m)^*+G(3,2m)\Bigr).
\end{split}
\ee
Using the standard result for the Gauss sum
\be
G(3,m)=\left\{ \begin{array}{ll}
\left\{ \begin{array}{ll}
(-1)^{(m-4)/4}(1-\I )\sqrt{m}, & m=0\!\!\mod 4
\\
(-1)^{(m-1)/4}\sqrt{m}, & m=1\!\!\mod 4,
\\
0, & m=2\!\!\mod 4,
\\
(-1)^{(m-3)/4} \I \sqrt{m}, & m=3\!\!\mod 4,
\end{array}\right.
\qquad \gcd(3,m)= 1,
\\
3\,G(1,m/3),
\hspace{5cm}
\begin{array}{c}
\gcd(3,m)=3,
\\
m=0\!\!\mod 3,
\end{array}
\end{array}\right.
\ee
along with \eqref{G1m}, we can rewrite \eqref{chiSTeven0} as
\begin{equation}
\begin{split}
\chi_M(ST)
&=\frac{e^{\frac{\pi\I}{3}\,mr^2}}{4\sqrt{\I m}}
\left(\sqrt{2m}(1-\I ) +\left\{ \begin{array}{ll}
\sqrt{6m}(1+\I ), &\quad  m=0\!\!\mod 6
\\
\sqrt{2m}(1-\I ), &\quad  m=2\!\!\mod 6
\\
-\sqrt{2m}(1-\I ), &\quad  m=4\!\!\mod 6
\end{array}\right. \right)
\\&
=e^{\frac{\pi\I}{3}\,mr^2} \left\{ \begin{array}{ll}
e^{-\frac{\pi \I}{6}}, &\quad  m=0\!\!\mod 6
\\
e^{-\frac{\pi \I }{2}}, &\quad  m=2\!\!\mod 6
\\
0, &\quad  m=4\!\!\mod 6
\end{array}\right.
=\left\{ \begin{array}{ll}
e^{-\frac{\pi \I}{6}}, &\quad  m=0\!\!\mod 6,
\\
e^{\frac{\pi \I }{6}}, &\quad  m=2\!\!\mod 6,
\\
0, &\quad  m=4\!\!\mod 6,
\end{array}\right.
\end{split}
\label{chiSTeven}
\end{equation}
where in the last step we used that for $m=2 \!\!\mod 6$ one has
$\Nr=\pm 1,\pm 2 \!\!\mod 6$.

For $m$ odd, $\Nr$ is necessarily odd as well. We then have
\begin{equation}
\begin{split}
\chi_M(ST)&=\frac{(-1)^{\cI_\Nr}}{\sqrt{\I m}}\, e^{\frac{\pi\I}{4}\,mr^2 +\frac{\pi\I}{12}\, rc_2}
\left(1 +\sum_{\nu=1}^{(m-1)/2} (-1)^{\Nr\nu }\( e^{-\frac{\pi \I  \nu^2}{m}}+e^{\frac{3\pi \I \nu^2}{m}}\) \right).
\end{split}
\end{equation}
Substitution of $\cI_\Nr$ (\ref{chiODN}) and evaluation for low values of $m$ suggests that this can be further simplified to
\begin{equation}
\begin{split}
\chi_M(ST)&=e^{\frac{\pi\I}{12}\,m(r^2-1)}
\left\{ \begin{array}{ll}
e^{-\frac{\pi\I}{6}}, &\quad  m=1\!\!\mod 6,
\\
e^{\frac{\pi\I}{6}}, & \quad m=3\!\!\mod 6,
\\
0, & \quad  m=5\!\!\mod 6.
\end{array}\right.
\end{split}
\label{chiST0}
\end{equation}
Note that since $\Nr$ divides $m$, $\Nr=\pm 1 \!\!\mod 6$ if $m=1 \!\!\mod 6$,
while $\Nr=\pm 1$ or $3\!\!\mod 6$ if $m=3\!\!\mod 6$.
Therefore, $m(r^2-1)\in 24\IZ$ and we arrive at a result similar to \eqref{chiSTeven}
\begin{equation}
\chi_M(ST)=
\left\{ \begin{array}{ll}
e^{-\frac{\pi\I}{6}}, &\quad  m=1\!\!\mod 6,
\\
e^{\frac{\pi\I}{6}}, & \quad m=3\!\!\mod 6,
\\
0, & \quad  m=5\!\!\mod 6.
\end{array}\right.
\label{chiST}
\end{equation}

Substituting  \eqref{chiS}, \eqref{chiSTeven} and \eqref{chiST} into \eqref{defAsep} for
$w=-3/2$, we arrive at the final result for the elliptic contribution,  assuming that \eqref{chiST0} is indeed true,
\begin{equation}
\label{AsApAe}
\begin{split}
&A_{\rm e}= \left\{\begin{array}{ll} \frac{1}{4} (-1)^{\cI_\Nr}, &\quad  m=0,1 \!\!\mod 4,
\\
0, &\quad  m=2,3 \!\!\mod 4, \end{array}\right.
\qquad
+ \left\{\begin{array}{ll} -\frac{1}{3},&\quad  m=2,3\!\!\mod 6,
\\
0, &\quad  m=0,1,4,5\!\!\mod 6. \end{array}\right.
\end{split}
\end{equation}
Inserting this result in \eqref{dimCuspM}, one finds  the number of constraints on the polar terms,
\begin{equation}
\label{numconstr}
C_{r}(\CY)=
\frac{5d}{24}
-\frac{1}{2}\,\cI(M)-\sum_{\mu=0}^{d-1} ((\Delta_\mu))
+ \frac{1}{4}\, (-1)^{\cI_\Nr} \left( \delta_{m}^{(4)} +  \delta_{m-1}^{(4)} \right)
-\frac13 \left( \delta_{m-2}^{(6)} +  \delta_{m-3}^{(6)} \right),
\end{equation}
where we recall that $d=\lceil \frac{m+1}{2}\rceil$, $\delta_x^{(n)}$ is the mod-$n$ Kronecker delta
defined in \eqref{defdelta} and $\Delta_\mu$ is given by \eqref{Deltamu}.
The dimension of the space $\scM_\Nr(\CY)$ is obtained
by subtracting the number of constraints \eqref{numconstr} from the number of polar terms \eqref{numpolterms}:
\be
\dim \scM_\Nr(\CY)=
\sum_{\mu=0}^{d-1} \Delta_\mu
+\frac{7d}{24}
- \frac{1}{4} (-1)^{\cI_\Nr} \left( \delta_{m}^{(4)} +  \delta_{m-1}^{(4)} \right)
+\frac13 \left( \delta_{m-2}^{(6)} +  \delta_{m-3}^{(6)} \right).
\label{dimVV}
\ee
In particular, the dimension of $\scM_\Nr(\CY)$ grows proportionally to $m^2\Nr^2=\kappa^2\Nr^4$, while
the number of constraints grows at most linearly in $m=\kappa\Nr$.
In Table \ref{table2} we record the number of polar terms $n_r(\CY)$
and constraints $C_r(\CY)$ for rank up to 10 (see Table \ref{table1} for $r=1,2$).

\begin{table} [t]
\begin{centering}
$$
\begin{array}{|l|rr|rr|rr|rr|rr|rr|rr|rr|}
\hline
\CY & n_3 & C_3 & n_4 & C_4 & n_5 & C_5 & n_6 & C_6 & n_7 & C_7 & n_8 & C_8 & n_9 & C_9 & n_{10} & C_{10}  \\
\hline
 X_{5} & 96 & 1 & 241 & 3 & 475 & 2 & 923 & 4 & 1549 & 4 & 2595 & 6 & 3928 & 6 & 5961 & 5 \\
 X_{6} & 44 & 1 & 105 & 0 & 197 & 2 & 378 & 3 & 608 & 3 & 1014 & 0 & 1497 & 3 & 2283 & 4 \\
 X_{8} & 32 & 1 & 65 & 1 & 117 & 1 & 203 & 0 & 333 & 2 & 519 & 3 & 774 & 3 & 1121 & 3 \\
 X_{10} & 11 & 0 & 26 & 1 & 37 & 0 & 71 & 1 & 98 & 1 & 165 & 1 & 217 & 1 & 336 & 1 \\
 X_{4,2} & 126 & 0 & 312 & 0 & 659 & 1 & 1254 & 1 & 2192 & 2 & 3600 & 1 & 5606 & 3 & 8370 & 2
   \\
 X_{4,4} & 65 & 0 & 159 & 3 & 322 & 3 & 598 & 0 & 1033 & 4 & 1681 & 4 & 2591 & 1 & 3855 & 6
   \\
 X_{6,2} & 77 & 1 & 177 & 3 & 349 & 2 & 637 & 0 & 1084 & 2 & 1749 & 4 & 2679 & 3 & 3960 & 6
   \\
 X_{6,4} & 25 & 0 & 55 & 1 & 103 & 2 & 182 & 0 & 304 & 1 & 483 & 3 & 726 & 0 & 1066 & 3 \\
 X_{6,6} & 8 & 0 & 20 & 1 & 30 & 1 & 59 & 1 & 84 & 1 & 145 & 1 & 194 & 1 & 306 & 1 \\
 X_{3,3} & 237 & 3 & 627 & 1 & 1339 & 5 & 2650 & 7 & 4625 & 8 & 7770 & 2 & 12041 & 9 & 18292
   & 12 \\
 X_{4,2} & 208 & 0 & 525 & 4 & 1125 & 6 & 2150 & 1 & 3793 & 6 & 6254 & 9 & 9768 & 2 & 14630 &
   10 \\
 X_{3,2,2} & 399 & 3 & 1050 & 1 & 2325 & 6 & 4551 & 2 & 8127 & 9 & 13524 & 4 & 21285 & 11 &
   32025 & 5 \\
 X_{2,2,2,2} & 650 & 1 & 1766 & 9 & 3970 & 10 & 7840 & 4 & 14106 & 14 & 23581 & 15 & 37230 & 7
   & 56171 & 20 \\
   \hline
   \end{array}
$$
\caption{The number of polar terms $n_r(\CY)$ and the number of constraints $C_r(\CY)$ for $3\leq r\leq 10$
\label{table2}}
\end{centering}
\end{table}

\section{Generalized Hecke operator}
\label{sec-Hecke}

In this appendix we show how one can construct VV mock modular forms $\Gi{\kappa}$ from a given VV mock modular form $\Gi{1}$.
These modular forms are defined by the condition \eqref{relGTh},
which fixes the form of the completion, where $\Thi{\kappa}$ is given in \eqref{defRmu}.
Essentially, the difference between cases with different $\kappa$ is the representation that the modular forms belong to.
It is characterized by the multiplier system \eqref{STusual} and
is known as the Weil representation associated with the even integral lattice $\Lambda=2\kappa\IZ$ with the discriminant group
$\Lambda^*/\Lambda=\IZ_{2\kappa}$. Thus, we simply need to find an operator which maps modular forms from one Weil representation to another.
In addition, we also need to ensure that it properly maps the anomalies for our mock modular forms captured by the functions $\Thi{\kappa}$.
This gives rise to the two conditions spelled out in \S\ref{subsec-explicit}.

It turns out that an operator satisfying the first condition
has already been constructed in \cite{Bouchard:2016lfg,Bouchard:2018pem}:
\begin{theorem}[\cite{Bouchard:2018pem}]
Let $\Lambda$ be a lattice of signature $(b^+,b^-)$ with bilinear form $(\, \cdot\,, \,\cdot\,)$,
$A=\Lambda^*/\Lambda$, and $\Lambda(\kappa)$ is
the same lattice but rescaled bilinear form
$(\, \cdot\,, \,\cdot\,)_\kappa=\kappa(\, \cdot\,, \,\cdot\,)$.
Let $\phi_{\lambda\in A}$ be a VV modular form of weight $(w,\bw)$ and multiplier system
\be
\begin{split}
M_{\lambda\lambda'}(T)
=&\,
e^{\pi\I \lambda^2}\,\delta_{\lambda\lambda'},
\\
M_{\lambda\lambda'}(S)
=&\, \frac{1}{\sqrt{|A|}}\,
e^{-\frac{\pi\I}{4}\, (b^+-b^-) -2\pi\I(\lambda,\lambda')}.
\end{split}
\label{ST-hecke}
\ee
Then the vector
\be
(\cT_\kappa[\phi])_\mu(\tau)=\frac{1}{\kappa} \sum_{a,d>0 \atop ad=\kappa}
\left(\frac{\kappa}{d}\right)^{w+\bw+\hf(b^++b^-)} \delta_\kappa(\mu,d)
\sum_{b=0}^{d-1}\, e^{-\pi\I \, \frac{b}{a}\,\mu^2}
\phi_{d\mu} \(\frac{a\tau+b}{d}\),
\label{defHecke}
\ee
with $\mu\in A(\kappa)$
and
\be
\delta_\kappa(\mu,d)=\left\{ \begin{array}{ll}
1\ & \mbox{if } \mu\in A(d)\subseteq A(\kappa),
\\
0 \ & \mbox{otherwise,}
\end{array}\right.
\ee
is a VV modular form of the same weight and multiplier system \eqref{ST-hecke} where the bilinear form
is replaced by the rescaled one.

\end{theorem}

In our case we take the rescaled bilinear form to be $(k_1, k_2)_\kappa=-2\kappa k_1k_2$,
so that its signature is $(b^+,b^-)=(0,1)$. Then we replace $\mu$ by $-\frac{\mu}{2\kappa}$
with $\mu\in\{0,\dots,2\kappa-1\}$, so that $\mu^2$ becomes $-\frac{\mu^2}{2\kappa}$.
After these substitutions and choosing the weight $(w,\bw)=(3/2,0)$, the action of the generalized Hecke operator becomes
\be
(\cT_\kappa[\phi])_\mu(\tau)=\kappa\sum_{a,d>0 \atop ad=\kappa}d^{-2}\sum_{b=0}^{d-1}
\delta^{(1)}_{\mu/a} \, e^{\frac{\pi\I b}{2 a\kappa}\,\mu^2}
\phi_{\mu/a} \(\frac{a\tau+b}{d}\),
\label{defHecke-exp}
\ee
and the multiplier system \eqref{ST-hecke} coincides with the one in \eqref{STusual}.
This agreement justifies the application of the above theorem to our problem.
More precisely, acting by $\cT_\kappa$ on \eqref{relGTh} with $\kappa=1$, we obtain
\be
(\cT_\kappa[\whGi{1}])_\mu=(\cT_\kappa[\Gi{1}])_\mu+(\cT_\kappa[\Thi{1}])_\mu.
\label{relGTh-cT}
\ee
The theorem ensures that the l.h.s. is a VV modular form so that it can be identified (up to a constant factor $c_\kappa$)
with $\whGi{\kappa}_\mu$.
Provided the last term on the r.h.s. coincides with $c_\kappa\Thi{\kappa}_\mu$,
the first term can then be identified with $c_\kappa\Gi{\kappa}_\mu$
and the operator generating the solution \eqref{GkapH} can be taken to be $\cT'_\kappa=c_\kappa^{-1}\cT_\kappa$.

Let us evaluate the action of the Hecke operator on $\Thi{1}$ explicitly.
Substituting \eqref{defRmu} into \eqref{defHecke-exp}, one finds
\be
\begin{split}
(\cT_\kappa[\Thi{1}])_\mu=&\, \frac{\kappa}{8\pi}\sum_{a,d>0 \atop ad=\kappa}
\sum_{b=0}^{d-1}\delta^{(1)}_{\mu/a} \, e^{\frac{\pi\I b}{2 a\kappa}\,\mu^2}
\sum_{k\in 2\IZ+\frac{d\mu}{\kappa}}\frac{|k|}{d^2}\,\beta_{\frac{3}{2}}\!\(\frac{a\tau_2}{d}\, k^2\)
e^{-\frac{\pi\I }{2}\(\frac{a\tau+b}{d}\) k^2 }
\\
=&\, \frac{1}{8\pi}\sum_{a|\kappa,\mu}
\sum_{k\in 2\IZ+\frac{\mu}{a}}|ak|\,\beta_{\frac{3}{2}}\!\(\frac{\tau_2}{\kappa}\, (ak)^2\)
e^{-\frac{\pi\I \tau}{2\kappa}\, (ak)^2 }
\, \frac{a}{\kappa}\sum_{b=0}^{\frac{\kappa}{a}-1} e^{\frac{\pi\I b}{2 a\kappa}\,(\mu^2-(ak)^2)}.
\end{split}
\label{action1}
\ee
Representing $k=2\eps+\mu/a$ where $\eps\in\IZ$, the last factor becomes
\be
\frac{a}{\kappa}\sum_{b=0}^{\frac{\kappa}{a}-1} e^{\frac{\pi\I b}{2 a\kappa}\,(\mu^2-(ak)^2)}
=\frac{a}{\kappa}\sum_{b=0}^{\frac{\kappa}{a}-1} e^{-2\pi\I b\,\frac{a}{\kappa}\,(\eps^2+\eps\mu/a)}
=\delta^{(\kappa/a)}_{\eps(\eps+\mu/a)}.
\label{delta}
\ee
Thus, we arrive at the constraint
\be
\eps(\eps+\mu/a)=0\mod \kappa/a.
\label{epszero}
\ee
Let us denote $\cS(\mu,a)$ the set of its integer solutions in the range $0\le \eps<\kappa/a$ and note that $\cS(\mu,a)+n\kappa/a$
also solves \eqref{epszero} for any $n\in\IZ$.
Therefore, \eqref{action1} can be rewritten as
\be
\begin{split}
(\cT_\kappa[\Thi{1}])_\mu=&\, \frac{1}{8\pi}\sum_{a|\kappa,\mu}\sum_{\eps\in \cS(\mu,a)}
\sum_{k\in \frac{2\kappa}{a}\IZ+\frac{\mu}{a}+2\eps}|ak|\,\beta_{\frac{3}{2}}\!\(\frac{\tau_2}{\kappa}\, (ak)^2\)
e^{-\frac{\pi\I \tau}{2\kappa}\, (ak)^2 }
\\
=&\, \sum_{a|\kappa,\mu}\sum_{\eps\in \cS(\mu,a)}\Thi{\kappa}_{\mu+2a\eps}\, .
\end{split}
\label{HeckeTh}
\ee
To proceed further, we need to find $\cS(\mu,a)$ explicitly. Due to the invariance under
$\mu\to -\mu$ and $\mu\to\mu+2\kappa$, it suffices to consider $\mu=0,\dots,\kappa$.

First, let us consider the case where $\kappa$ is a prime number.
Then for $\mu=0$ and $\mu=\kappa$, $a$ takes two values, $1$ and $\kappa$,
and in both cases the only solution of \eqref{epszero} is $\eps=0$.
So \eqref{HeckeTh} results in $2\Thi{\kappa}_0$. On the other hand, for $1\le \mu<\kappa$, $a=1$ and the condition \eqref{epszero}
has two solutions: $\eps=0$ and $\eps=-\mu \mod \kappa$.
Thus, \eqref{HeckeTh} results in $\Thi{\kappa}_\mu+\Thi{\kappa}_{-\mu}=2\Thi{\kappa}_\mu$.
Hence, for all $\mu$ one obtains
\be
(\cT_\kappa[\Thi{1}])_\mu=2\Thi{\kappa}_\mu.
\label{Thkappa-prime}
\ee
Thus, one may simply take $\cT'_\kappa=\cT_\kappa/2$ in this case.

When $\kappa$ is non-prime, we were not able to find a general solution of \eqref{epszero}.
Nonetheless, it is straightforward to analyze small values of $\kappa$ case-by-case,
including the values $\kappa=4,6,8,9,12,16$ appearing in Table \ref{table1}.
Rather than listing the solutions of \eqref{epszero} in each case,  we shall simply
state the result of applying the Hecke operator on $\Thi{1}$ \eqref{HeckeTh}:
\begin{subequations}
\bea
(\cT_4[\Thi{1}])_\mu &=& 2\Thi{4}_\mu+\delta^{(2)}_{\mu}(\Thi{4}_\mu+\Thi{4}_{\mu+4}),
\label{res4}\\
(\cT_6[\Thi{1}])_\mu &=& 4\Thi{6}_\mu-2\delta^{(2)}_{\mu+1}(\Thi{6}_\mu-\Thi{6}_{\mu+6}).
\label{res6}
\\
(\cT_8[\Thi{1}])_\mu &=& 2\Thi{8}_\mu+2\delta^{(2)}_{\mu}(\Thi{8}_\mu+\Thi{8}_{\mu+8}),
\\
(\cT_9[\Thi{1}])_\mu &=& 2\Thi{9}_\mu+\delta^{(3)}_{\mu}(\Thi{9}_\mu+\Thi{9}_{\mu+6}+\Thi{9}_{\mu+12}),
\\
(\cT_{12}[\Thi{1}])_\mu &=& 4\Thi{12}_\mu+2\delta^{(2)}_{\mu}(\Thi{12}_\mu+\Thi{12}_{\mu+12})
\label{res12}\\
&&
-2\[\delta^{(4)}_{\mu+2}(\Thi{12}_\mu-\Thi{12}_{\mu+12})+\delta^{(6)}_{\mu+1}(\Thi{12}_\mu-\Thi{12}_{\mu+16})
+\delta^{(6)}_{\mu+5}(\Thi{12}_\mu-\Thi{12}_{\mu+8})\],
\nn\\
(\cT_{16}[\Thi{1}])_\mu &=& 2\Thi{16}_\mu+2\delta^{(2)}_{\mu}(\Thi{16}_\mu+\Thi{16}_{\mu+16})
+\delta^{(4)}_{\mu}(\Thi{16}_\mu+\Thi{16}_{\mu+8}+\Thi{16}_{\mu+16}+\Thi{16}_{\mu+24}).
\eea
\label{cTTh-kap}
\end{subequations}
Thus, unlike for prime $\kappa$, we cannot just take $\cT'_\kappa$ to be proportional to $\cT_\kappa$.
However, upon closer examination one can still find an operator $\cT'_\kappa$
that satisfies all the required conditions
when $\kappa=4,8,9,16$ (or more generally, when $\kappa$ is a prime power).
Indeed, in those cases,
each of the additional terms  in \eqref{cTTh-kap} can be shown to transform
in the correct representation due to the following proposition
(which is a variant of Proposition 1 from \cite{Alexandrov:2020bwg}):
\begin{proposition}
Let $\theta_\mu$ ($\mu=0,\dots,2\kappa-1$) be a VV modular form transforming with the multiplier system \eqref{STusual}
and $d\in\IN$ : $d^2$ divides $\kappa$.
Then the vector with components
\be
(\Sigma_{\kappa,d}[\theta])_\mu=\delta^{(d)}_\mu\sum_{n=0}^{d-1}\theta_{\mu+ 2n\kappa/d}
\label{defvector}
\ee
transforms according to the same representation.
\label{propSigma}
\end{proposition}
\begin{proof}
First, we verify the T-transformation. Acting on each term in the sum \eqref{defvector},
it produces the following phase factor
\be
e^{-\frac{\pi\I}{2\kappa}\(\mu+ \frac{2n\kappa}{d}\)^2}=e^{-\frac{\pi\I}{2\kappa}\mu^2-2\pi\I\(\frac{n\mu}{d} + \frac{n^2\kappa}{d^2}\)}
=e^{-\frac{\pi\I}{2\kappa}\mu^2},
\ee
where we used that $d$ divides $\mu$ and $d^2$ divides $\kappa$.
The result reproduces the phase factor in \eqref{STusual}.
To check the S-transformation, we evaluate
\bea
&&
\frac{\sqrt{\I}}{\sqrt{2\kappa}}\,\delta^{(d)}_\mu
\sum_{n=0}^{d-1}\sum_{\nu=0}^{2\kappa-1} e^{\frac{\pi\I\nu}{\kappa}\(\mu+ \frac{2n\kappa}{d}\)}\,
\theta_{\nu}
=\frac{\sqrt{\I}}{\sqrt{2\kappa}}\, \sum_{m=0}^{d-1} e^{2\pi\I \mu \frac{m}{d}}
\sum_{\nu=0}^{2\kappa-1} e^{\frac{\pi\I\nu}{\kappa}\mu}\,\delta^{(d)}_\nu\,
\theta_{\nu}
\nn\\
&=& \frac{\sqrt{\I}}{\sqrt{2\kappa}}\,
\sum_{\nu=0}^{2\kappa-1}\delta^{(d)}_\nu\sum_{m=0}^{d-1}
e^{\frac{\pi\I\nu}{\kappa}\mu}\, \theta_{\nu-2m\kappa/d}
=\sum_{\nu=0}^{2\kappa-1}\Mi{\kappa}_{\mu\nu}(S)(\Sigma_{\kappa,d}[\theta])_\nu,
\eea
which confirms the correct transformation.
\end{proof}
Therefore, we can consider each equation \eqref{cTTh-kap} as a system of linear equations on the quantities $\Thi{\kappa}_\mu$
to be expressed through $(\cT_\kappa[\Thi{1}])_\mu$.
As a result, we obtain
\be
\Thi{\kappa}_\mu=(\cT'_\kappa[\Thi{1}])_\mu,
\ee
where
\begin{subequations}
\bea
(\cT'_\kappa[\phi])_\mu&=&\hf\,(\cT_\kappa[\phi])_\mu, \qquad \kappa\mbox{ --- prime},
\\
(\cT'_4[\phi])_\mu&=&\hf\,(\cT_4[\phi])_\mu-\frac18\(\Sigma_{4,2}[\cT_4[\phi]]\)_\mu,
\label{defG4}
\\
(\cT'_8[\phi])_\mu&=&\hf\,(\cT_8[\phi])_\mu-\frac16 \(\Sigma_{8,2}[\cT_8[\phi]]\)_\mu,
\label{defG8}
\\
(\cT'_9[\phi])_\mu&=&\hf\,(\cT_9[\phi])_\mu-\frac{1}{10}\(\Sigma_{9,3}[\cT_9[\phi]]\)_\mu,
\label{defG9}
\\
(\cT'_{16}[\phi])_\mu&=&\hf\,(\cT_{16}[\phi])_\mu
-\frac16\(\Sigma_{16,2}[\cT_{16}[\phi]]\)_\mu
-\frac{1}{60}\(\Sigma_{16,4}[\cT_{16}[\phi]]\)_\mu.
\label{defG16}
\eea
\label{defTp}
\end{subequations}
So in general for $\kappa=p^m$ where $p$ is a prime number we expect that
\be
(\cT'_{\kappa}[\phi])_\mu=\hf\sum_{n=0}^{\lfloor m/2\rfloor} c_{\kappa,n}\(\Sigma_{\kappa,p^n}[\cT_{\kappa}[\phi]]\)_\mu
\ee
where $c_{\kappa,0}=1$ and $c_{\kappa,n}$ with $n>0$ are some negative rational numbers.
Due to Proposition \ref{propSigma}, all terms in the sum transform with the same multiplier system \eqref{STusual}.
Hence, $\cT'_\kappa$ is the operator satisfying all our requirements
and allowing the identification \eqref{GkapH}.

Finally, let us consider the case when $\kappa=6$ or 12.
Although it can be checked that all additional terms in \eqref{res6} and \eqref{res12}
do transform with the proper multiplier system \eqref{STusual}, it turns out that these equations cannot be solved for $\Thi{\kappa}_\mu$
because \eqref{res6} does not depend on $\Thi{6}_1-\Thi{6}_7$, while \eqref{res12} does not involve
$\Thi{12}_1-\Thi{12}_7$ and $\Thi{12}_5-\Thi{12}_{11}$.
Thus, it seems that when $\kappa$ is a product of different prime integers,
our approach based on the generalized Hecke operator $\cT_\kappa$ does not work
and a more complicated construction is required.

\section{Generating functions for unit D4-brane charge}
\label{sec-genfun1}

In this section we provide tables of the rank 1 DT invariants $DT(Q,n)$ which enter in the Ansatz \eqref{hpolar}
for the polar part of the generating series $h_{1,\mu}$, and for the 10 models
in which a VV modular form with the required polar part exists,
give the generating series expressed as in \eqref{decomp-modform} and their first few terms in the $\q$-expansion.
In presenting these results, we underline the polar terms and put the number  $n$  of D0-branes responsible for
each polar term as a subscript.
We also discuss the remaining three models, in particular, how their polar terms can be corrected to allow for a solution.
The invariants $DT(Q,n)$ are computed from the GV invariants listed in \cite{Huang:2006hq}
(with some corrections kindly pointed out by the authors).
All computations can be found in an ancillary Mathematica notebook available on arXiv.

\subsection*{$\mathbf{X_5}$}

The lowest DT invariants are as follows (extending the table in \cite{Collinucci:2008ht}):
\be
{\scriptsize
\begin{array}{|c|cccccccc}
\hline
Q \backslash n & -2 & -1 & 0 & 1 & 2 & 3 \\ \hline
0  & 0 & 0 & 1 & 200 & 19500 & 1234000 \\
1 & 0 & 0 & 0 & 2875 & 569250 & 54921125 \\
2 & 0 & 0 & 0 & 609250 & 124762875 & 12448246500 \\
3 & 0 & 0 & 609250 & 439056375 & 76438831000 & 7158676736750 \\
4 & 8625 & 2294250 & 4004590375 & 1010473893000 & 123236265797125 & 9526578133835000 \\
\hline
\end{array}
}\nn
\ee
The generating function is found to be\footnote{Here and below it is understood that the argument $z$
of the theta function $\vths{\kappa,1}_{\mu}$ defined in \eqref{Vignerasth} is set to zero after taking derivative.
The Fourier expansion is given only for the components with $0\le \mu\le \kappa/2$ as the other components are fixed by
the symmetry $h_{1,\mu}=h_{1,-\mu}$.}
\be
\begin{split}
h_{1,\mu}=&\, -\frac{1}{2\pi\eta^{70}} \[-\frac{222887 E_4^8+1093010 E_4^5 E_6^2+177095 E_4^2 E_6^4}{35831808}
\right.
\\
&\,
+\frac{25 \(458287 E_4^6 E_6+967810 E_4^3 E_6^3+66895 E_6^5\)}{53747712}\, D
\\
&\, \left.
+\frac{25 \(155587 E_4^7+1054810 E_4^4 E_6^2+282595 E_4 E_6^4\)}{8957952}\, D^2\]\p_z\vths{5,1}_{\mu},
\end{split}
\ee
and has the following expansion (which agrees with the results in \cite[Eq.(3.10)]{Gaiotto:2006wm} and \cite[Eq.(2.3)]{Gaiotto:2007cd}):
\be
\begin{split}
h_{1,0} =&\, \q^{-\frac{55}{24}} \left( \underline{5_0 - 800_1 q + 58500_2 \q^2} + 5817125 \q^3 + 75474060100 \q^4
+ \dots \right),
\\
h_{1,1} =&\, \q^{-\frac{55}{24}+\frac35}  \left( \underline{0_0+8625_1 q}- 1138500 \q^2 + 3777474000 \q^3
+ 3102750380125 \q^4 + \dots\right),
\\
h_{1,2} =&\, \q^{-\frac{55}{24}+\frac25}
\left(\underline{0_{-1}+0_0 q}  -1218500 \q^2 + 441969250 \q^3 + 953712511250 \q^4
+ \dots\right).
\end{split}
\ee

\subsection*{$\mathbf{X_6}$}

The lowest DT invariants are as follows:
\be
{\scriptsize\begin{array}{|c|ccccccc}
 \hline
Q \backslash n & -3 & -2 & -1 & 0 & 1 & 2 & 3 \\  \hline
0& 0 & 0 & 0 & 1 & 204 & 20298 & 1311584 \\
1& 0 & 0 & 0 & 0 & 7884 & 1592568 & 156836412 \\
2& 0 & 0 & 0 & 7884 & 7636788 & 1408851522 & 136479465324 \\
3& 6 & 1836 & 266526 & 169502712 & 43151185260 & 5487789706776 & 440955379766460 \\
4& -47304 & -24852636 & 6684091812 & 3616211898459 & 597179528504352 & 56820950585055180 &
   3715523804755065780 \\  \hline
\end{array}} \nn
\ee
The generating function is found to be
\be
\begin{split}
h_{1,\mu}=&\,- \frac{1}{2\pi\eta^{54}} \[\frac{7 E_4^6+58 E_4^3 E_6^2+7 E_6^4}{216}
+\frac{5 E_4^4 E_6+3 E_4 E_6^3}{2}\, D\]\p_z\vths{3,1}_{\mu},
\end{split}
\ee
and has the following expansion (which agrees with \cite[Eq.(2.7)]{Gaiotto:2007cd}, up to overall sign):
\be
\begin{split}
\hspace{-0.2cm}
h_{1,0} =&\,  \q^{-\frac{15}{8}}
\left( \underline{-4_0 +612_1 \q }-  40392 \q^2 +146464860 q^3 +66864926808 \q^4 + \dots\right),
\\
\hspace{-0.2cm}
h_{1,1} =&\,  \q^{-\frac{15}{8}+\frac{2}{3}} \left( \underline{0_0 - 15768_1 \q}
 +7621020 \q^2 + 10739279916 \q^3 +1794352963536 \q^4 +\dots
\right).
\end{split}
\ee

\subsection*{$\mathbf{X_8}$}

The lowest DT invariants are as follows:
\be
\hspace*{-1.5cm}
{\scriptsize
\begin{array}{|c|cccccc} \hline
Q\backslash n & -2 & -1 & 0 & 1 & 2 & 3 \\ \hline
0&   0 & 0 & 1 & 296 & 43068 & 4104336 \\
1& 0 & 0 & 0 & 29504 & 8674176 & 1253300416 \\
2&  6 & 2664 & 564332 & 204456696 & 45540821914 & 6127608486208 \\
3& -177024 & -69481920 & 8775447296 & 6313618655104 & 1225699503521536 &
   141978726005461504 \\ \hline
   \end{array}}
\nn
\ee
The generating function is found to be
\be
\begin{split}
h_{1,\mu}=&\, \frac{1}{\eta^{52}} \[\frac{103 E_4^6+1472 E_4^3 E_6^2+153 E_6^4}{5184}
+\frac{503 E_4^4 E_6+361 E_4 E_6^3}{108}\, D\]\vths{2,1}_{\mu},
\end{split}
\ee
and has the following expansion (which agrees with \cite[Eq.(2.11)]{Gaiotto:2007cd}, up to overall sign):
\be
\begin{split}
\hspace{-0.2cm}
h_{1,0}=&\, \q^{-\frac{46}{24}} \left( \underline{-4_0 + 888_1 \q} - 86140 \q^2 +132940136 \q^3 +86849300500 \q^4
+\dots\right),
\\
\hspace{-0.2cm}
h_{1,1}=&\, \q^{-\frac{46}{24}+\frac34} \left( \underline{0_0 - 59008_1  \q} + 8615168 \q^2 +21430302976 \q^3
+3736977423872 \q^4 +\dots\right).
\end{split}
\ee

\subsection*{$\mathbf{X_{10}}$}

The lowest DT invariants are as follows:
\be
{\scriptsize
\begin{array}{|c|ccccccc} \hline
Q \backslash n & -3 & -2 & -1 & 0 & 1 & 2 & 3 \\ \hline
0 & 0 & 0 & 0 & 1 & 288 & 40752 & 3774912 \\
1& 0 & 0 & 3 & 1150 & 435827 & 89103872 & 11141118264 \\
2 & -12 & -5181 & -1529746 & -64916198 & 40225290446 & 9325643249563 & 1112733511380100 \\
\hline
\end{array}
}\nn
\ee
The generating function is found to be$^\dagger$
\be
\label{X10gen}
\begin{split}
h_{1,0}=&\, \frac{541 E_4^4+1187 E_4 E_6^2}{576\, \eta^{35}}
\\
=&\, \q^{-\frac{35}{24}}\Bigl(
\underline{3_0-576_1 q}+271704 \q^2+ 206401533 \q^3+ 21593767647 \q^4
+\cdots\Bigr),
\end{split}
\ee
where we took into account that $\p_z\vths{1,1}_0(\tau,0)=-2\pi\eta^3(\tau)$.
This result agrees with \cite[Eq.(2.12)]{Gaiotto:2007cd}.\footnote{It was suggested in \cite{VanHerck:2009ww}
that the second polar coefficient should be modified to $-575$, but this suggestion was
not taken seriously in the initial version of the present work. As explained in the note on page \pageref{noteadded},
it is in fact confirmed by the mathematical results of \cite{Feyzbakhsh:2022ydn}, hence
we have marked \eqref{X10gen} with a $\dagger$.  The correct expansion can be found
in \cite[(5.10)]{VanHerck:2009ww}. \label{fooWyder}}

\subsection*{$\mathbf{X_{4,3}}$}

The lowest DT invariants are as follows:
\be
{\scriptsize
\begin{array}{|c|cccccccccc} \hline
 Q \backslash n  & -6 & -5 & -4 & -3 & -2 & -1 & 0 & 1 & 2 
  \\ \hline
0 & 0 & 0 & 0 &  0 & 0 & 0 & 1 & 156 & 11778 
 \\
1 & 0 & 0 & 0 & 0 & 0 & 0 & 0 & 1944 & 299376 
 \\
2 & 0 & 0 & 0 & 0 & 0 & 0 & 27 & 227772 & 36634842 
\\
3 & 0 & 0 & 0 &  0 & 0 & 0 & 161248 & 89961744 & 12314066208 
\\
4 & 0 & 0 & 0 &  0 & 81 & 240408 & 418646475 & 90148651920 & 9065616005898 
 \\
5 & 0 & 0 & 0 & 5832 & 1100304 & 3996193968 & 7007431566096 & 1058781525672312 & 79955621660025792
    \\
6 & 10 & 2496 & 275273 & -21407812 & 69458828969 & 32461114565928 & 5111995215726463 &
   460091731369849584 & 28020271480178497520 
    \\
   \hline
\end{array}
}
\nn
\ee
The generating function is found to be$^\dagger$
\be
\begin{split}
h_{1,\mu}=&\, \frac{1}{\eta^{72}} \[\frac{709709 E_4^7 E_6 - 3221146 E_4^4 E_6^3 - 1359283 E_4 E_6^5}{637009920}
\right.
\\
&\,
+\frac{1106929 E_4^8 + 5476894 E_4^5 E_6^2 + 604657 E_4^2 E_6^4}{26542080}\, D
\\
&\,
-\frac{58663 E_4^6 E_6 + 117682 E_4^3 E_6^3 + 7975 E_6^5}{73728}\, D^2
\\
&\, \left.
-\frac{62453 E_4^7 + 395798 E_4^4 E_6^2 + 94709 E_4 E_6^4}{46080}\, D^3\]\vths{6,1}_{\mu},
\end{split}
\ee
and has the following expansion:$^\dagger$
\be
\begin{split}
h_{1,0}=&\, \q^{-\frac94} \left( \underline{5_0 - 624_1 \q + 35334_2 \q^2}
+19017138 \q^3 + 74785371360 \q^4 +\dots \right),
\\
h_{1,1} =&\, \q^{-\frac94+\frac{7}{12}} \left( \underline{0_0+5832_1 \q}
-544806 \q^2 + 3919919670 \q^3 + 2506521890376 \q^4+\dots \right),
\\
h_{1,2} =&\, \q^{-\frac94+\frac13} \left( \underline{0_{-1}+81_0\q}
-455787\q^2 + 418792680 \q^3 + 589406281317 \q^4 +\dots \right),
\\
h_{1,3} =&\, \q^{-\frac94+\frac14} \left( \underline{0_{-2}+0_{-1}\q}
-322658 \q^2 + 154766856 \q^3 + 356674009104 \q^4+\dots \right).
\end{split}
\ee

\subsection*{$\mathbf{X_{4,4}}$}

The lowest DT invariants are as follows:
\be
{\scriptsize
\begin{array}{|c|cccccccc} \hline
 Q \backslash n & -4 & -3 & -2 & -1 & 0 & 1 & 2 & 3 \\ \hline
 0 & 0 & 0 & 0 & 0 & 1 & 144 & 10008 & 446304 \\
 1 & 0 & 0 & 0 & 0 & 0 & 3712 & 527104 & 36091776 \\
 2 & 0 & 0 & 0 & 0 & 1408 & 1185216 & 160488768 & 11145152320 \\
 3 & 0 & 0 & 0 & 3712 & 7495680 & 1728263936 & 172767389440 & 10375330097920 \\
 4 & 6 & 1296 & 112296 & 153732336 & 48667802732 & 6124054838960 & 444235976561624 & 21742669957124080 \\ \hline
\end{array}
} \nn
\ee
Although there is a modular constraint on polar terms, it turns out to be satisfied by our Ansatz due to
the following relation between the DT invariants
\be
\label{uncanny1}
DT(0,0)+\frac{3}{16}\, DT(0,1)-\frac{1}{32}\, DT(1,1) +\frac{1}{16}\, DT(2,0) = 0.
\ee
The resulting  generating function is found to be
\be
\begin{split}
h_{1,\mu}=&\, \frac{1}{\eta^{56}} \[\frac{319 E_4^5 E_6+113 E_4^2 E_6^3}{11664 }
-\frac{146 E_4^6+1025 E_4^3 E_6^2+125 E_6^4}{972}\, D
\right.
\\
&\,\left.
-\frac{566 E_4^4 E_6+298 E_4 E_6^3}{81}\, D^2\]\vths{4,1}_{\mu},
\end{split}
\ee
and has the following expansion:
\be
\begin{split}
\hspace{-0.2cm}
h_{1,0}=&\, \q^{-\frac{44}{24}} \left( \underline{-4_0 + 432_1 \q} - 10032 \q^2 + 148611456 \q^3 +53495321332 \q^4
+\dots\right),
\\
\hspace{-0.2cm}
h_{1,1} =&\, \q^{-\frac{44}{24}+\frac58} \left( \underline{0_0 - 7424_1  \q} +7488256 \q^2 +7149513728 \q^3
+1104027086592 \q^4 +\dots\right),
\\
\hspace{-0.2cm}
h_{1,2} =&\, \q^{-\frac{44}{24}+\frac12} \left( \underline{0_{-1} - 2816_0  \q} +2167680 \q^2 + 3503031296 \q^3
+619015800576 \q^4 +\dots\right).
\end{split}
\label{exp-hX62}
\ee

\subsection*{$\mathbf{X_{6,2}}$}

The lowest DT invariants are as follows:
\be
{\scriptsize
\begin{array}{|c|cccccccc} \hline
 Q \backslash n  &-4 & -3 & -2 & -1 & 0 & 1 & 2 & 3 \\ \hline
 0 & 0 & 0 & 0 & 0 & 1 & 256 & 32128 & 2633216 \\
 1 & 0 & 0 & 0 & 0 & 0 & 4992 & 1267968 & 157842048 \\
 2 & 0 & 0 & 0 & -4 & -1536 & 2129180 & 592221184 & 76687779936 \\
 3 & 0 & 0 & 0 & 14976 & 5071872 & 3527640064 & 784442776832 & 94963960029952 \\
 4 & 10 & 4096 & 810898 & 87634944 & 84783721868 & 25072077880832 & 3730330724940930 &
   357859766301860864 \\ \hline
\end{array}
} \nn
\ee
The generating function is found to be$^\dagger$
\be
\begin{split}
h_{1,\mu}=&\, \frac{1}{\eta^{68}} \[-\frac{994693 E_4^8+4317814 E_4^5 E_6^2+2152453 E_4^2 E_6^4}{161243136}
\right.
\\
&\,
+\frac{1974661 E_4^6 E_6+5095030 E_4^3 E_6^3+395269 E_6^5}{4478976}\, D
\\
&\, \left.
+\frac{738373 E_4^7+5203702 E_4^4 E_6^2+1522885 E_4 E_6^4}{559872}\, D^2\]\vths{4,1}_{\mu},
\end{split}
\ee
and has the following expansion:\symfootnote{As indicated in the note on page \pageref{noteadded},
the ansatz \eqref{hpolar} fails to give the correct polar terms for this model, so these results should not be trusted.
}
\be
\begin{split}
\hspace{-0.4cm}
h_{1,0}=&\,\q^{-\frac{56}{24}} \left( \underline{5_0 -1024_1 \q + 96384_2 \q^2}
-1082400 \q^3 + 87565497502 \q^4
+\dots\right),
\\
\hspace{-0.4cm}
h_{1,1} =&\, \q^{-\frac{56}{24}+\frac58} \left( \underline{0_0 + 14976_1  \q}  -1135328 \q^2 + 2168240416 \q^3 + 3646461843520 \q^4
+\dots\right),
\\
\hspace{-0.4cm}
h_{1,2} =&\, \q^{-\frac{56}{24}+\frac12} \left( \underline{16_{-1} -4608_0  \q} -5272444 \q^2 + 903979584 \q^3 + 2117148662336 \q^4
 +\dots\right).
\end{split}
\ee

\subsection*{$\mathbf{X_{6,4}}$}

The lowest DT invariants are as follows:
\be
{\scriptsize
\begin{array}{|c|ccccccc}  \hline
Q \backslash n &  -3 & -2 & -1 & 0 & 1 & 2 & 3 \\ \hline
0 & 0 & 0 & 0 & 1 & 156 & 11778 & 572416 \\
1 & 0 & 0 & 0 & 8 & 16800 & 2489232 & 182945216 \\
2 &  0 & 3 & 608 & 315828 & 71744924 & 7624177244 & 492335041044 \\
3&  -48 & -72184 & 26107984 & 10989768672 & 1476019954080 & 112615254328992 & 5813857713864192 \\ \hline
\end{array}} \nn
\ee
The generating function is found to be$^\dagger$
\be
\begin{split}
h_{1,\mu}=&\, \frac{1}{\eta^{40}} \[-\frac{509 E_4^3 E_6+139 E_6^3}{2592}
-\frac{233 E_4^4+415 E_4 E_6^2}{108}\, D\]\vths{2,1}_{\mu},
\end{split}
\ee
and has the following expansion:$^\dagger$
\be
\begin{split}
\hspace{-0.6cm}
h_{1,0}=&\,\q^{-\frac{34}{24}} \left( \underline{3_0 - 312_1 \q} +269343 \q^2 + 133568456 \q^3 + 12400947182 \q^4
+\dots\right),
\\
\hspace{-0.6cm}
h_{1,1} =&\, \q^{-\frac{34}{24}+\frac34} \left( \underline{-16_0} + 31904 \q + 36568960 \q^2 + 4364805376 \q^3 + 226013798816 \q^4+\dots\right).
\end{split}
\ee

\subsection*{$\mathbf{X_{6,6}}$}

The lowest DT invariants are as follows:
\be
{\scriptsize
\begin{array}{|c|ccccccc} \hline
Q \backslash n & -3 & -2 & -1 & 0 & 1 & 2 & 3 \\ \hline
0 &  0 & 0 & 0 & 1 & 120 & 6900 & 252400 \\
1 & 0 & 0 & 1 & 482 & 117445 & 10668592 & 545062022 \\
2 &   -6 & -1684 & 130808 & 67782432 & 7543637572 & 456342386980 & 18275307362778 \\
\hline
\end{array}
} \nn
\ee
The generating function is found to be
\be
\begin{split}
h_{1,0}=&\,
-\frac{2E_4 E_6}{\eta^{23}}
\\
=&\,  \q^{-\frac{23}{24}} \left( - \underline{2_0}
+482 \q + 282410 \q^2 + 16775192 \q^3
+ 460175332 \q^4+\dots \right).
\end{split}
\ee

\subsection*{$\mathbf{X_{3,3}}$}

The lowest DT invariants are as follows:
\be
{\scriptsize
\begin{array}{|c|ccccccc} \hline
Q \backslash n &  -3 & -2 & -1 & 0 & 1 & 2 & 3 \\ \hline
0 & 0 & 0 & 0 & 1 & 144 & 10008 & 446304 \\
1 &  0 & 0 & 0 & 0 & 1053 & 149526 & 10238319 \\
2&  0 & 0 & 0 & 0 & 52812 & 8053182 & 591031890 \\
3& 0 & 0 & 0 & 3402 & 6914214 & 1001912544 & 71961634872 \\
4 & 0 & 0 & 0 & 5520393 & 1937967282 & 225717793668 & 14749020131814 \\
5& 0 & 0 & 5520393 & 5626721862 & 1006811225253 & 88682916004956 & 4943255069504250 \\
6 &  10206 & 8383878 & 24521163804 & 6662846868372 & 768849614982540 & 52757172850669686 &
   2484705136566066336 \\ \hline
\end{array}
} \nn
\ee
Although there is a modular constraint on polar terms, it turns out to be satisfied by our Ansatz due to
the following relation between the DT invariants
\be
\label{uncanny2}
DT(0, 1) - \frac{2}{15}\, DT(0, 2) - \frac{92}{135}\,DT(1, 1) +\frac{1}{90}\, DT(1, 2)
+\frac{1}{90}\, DT(2, 1) -  \frac{1}{10}\, DT(3, 0) =0.
\ee
The resulting  generating function is found to be
\be
\begin{split}
h_{1,\mu}=&\, -\frac{1}{2\pi\eta^{90}} \[\frac{47723 E_4^9 E_6+25095 E_4^6 E_6^3-68943 E_4^3 E_6^5-3875 E_6^7}{107495424}
\right.
\\
&\,
+\frac{289326 E_4^{10}+415189 E_4^7 E_6^2-3458324 E_4^4 E_6^4-729839 E_4 E_6^6}{334430208}\, D
\\
&\,
+\frac{2261629 E_4^8 E_6+3219046 E_4^5 E_6^3-6371 E_4^2 E_6^5}{30965760}\, D^2
\\
&\,
-\frac{94271 E_4^9+1496733 E_4^6 E_6^2+1342665 E_4^3 E_6^4+52315 E_6^6}{5160960}\, D^3
\\
&\,\left.
-\frac{162167 E_4^7 E_6+300338 E_4^4 E_6^3+35159 E_4 E_6^5}{286720}\, D^4\]\p_z\vths{9,1}_{\mu},
\end{split}
\ee
and has the following expansion (which agrees with the results in  \cite[\S 2.5]{Gaiotto:2007cd}, up to overall sign):
\be
\begin{split}
\hspace{-0.4cm}
h_{1,0}=&\, \q^{-\frac{63}{24}}  \left( -\underline{6_0 + 720_1 \q - 40032_2 \q^2}
-678474 \q^3 + 30885198768 \q^4+\dots \right),
\\
\hspace{-0.4cm}
h_{1,1} =&\, \q^{-\frac{63}{24}+\frac59}  \left( \underline{0_0 - 4212_1 \q +448578_2 \q^2}
+374980104 \q^3 +2020724648442 \q^4+\dots \right),
\\
\hspace{-0.4cm}
h_{1,2} =&\, \q^{-\frac{63}{24}+\frac29}  \left( \underline{0_{-1} + 0_0 \q + 158436_1 \q^2}
-12471246 \q^3 +174600085086 \q^4+\dots \right),
\\
\hspace{-0.4cm}
h_{1,3} =&\, \q^{-\frac{63}{24}}  \left( \underline{0_{-2} + 0_{-1} \q +10206_0 \q^2}
-13828428 \q^3+24425287884 \q^4+\dots \right),
\\
\hspace{-0.4cm}
h_{1,4} =&\, \q^{-\frac{63}{24}+\frac89}  \left( \underline{0_{-2} + 0_{-1} \q}-11040786 \q^2
+6769752552 \q^3 + 17629606262268 \q^4 +\dots\right).
\end{split}
\ee

\subsection*{$\mathbf{X_{4,2}}$}

The lowest DT invariants are as follows:
\be
{\scriptsize
\begin{array}{|c|cccccc} \hline
Q \backslash n &   -2 & -1 & 0 & 1 & 2 & 3 \\ \hline
0 &  0 & 0 & 1 & 176 & 15048 & 831776 \\
1 &  0 & 0 & 0 & 1280 & 222720 & 18814720 \\
2&  0 & 0 & 0 & 92288 & 16876672 & 1497331072 \\
3 &   0 & 0 & 2560 & 16105728 & 2880650752 & 252911493632 \\
4 &  -8 & -2112 & 17161392 & 6933330304 & 961734375064 & 75838156759744 \\
\hline
\end{array}
} \nn
\ee
Our Ansatz \eqref{hpolar} implies the following polar terms
\be
\begin{split}
\hpol_{1,0} =&\, \q^{-\frac{8}{3}} \left( {-6_0 + 880_1 \q- 60192_2 \q^2}\right),
\\
\hpol_{1,1} =&\, \q^{-\frac{8}{3}+\frac{9}{16}} \left( {0_0-5120_1 \q+668160_2 \q^2} \right),
\\
\hpol_{1,2} =&\, \q^{-\frac{8}{3}+\frac14} \left( {0_{-1}+ 0_0 \q + 276864_1 \q^2}\right),
\\
\hpol_{1,3} =&\, \q^{-\frac{8}{3}+\frac1{16}} \left( {0_{-2}+0_{-1}\q+7680_0 \q^2} \right),
\\
\hpol_{1,4} =&\, \q^{-\frac{8}{3}} \left( {0_{-3}+32_{-2}\q-6336_{-1}\q^2}\right).
\end{split}
\label{polarX42}
\ee
However, they fail to satisfy the constraint imposed by modularity,
which would require that the DT invariants fulfill the relation
\be
\begin{split}
& DT(0,0) +\frac{5}{12}\, DT(0,1)-\frac{1}{6}\, DT(0,2)-\frac{29}{48}\,DT(1,1)+\frac{1}{64}\,DT(1,2)
\\
&
-\frac{9}{64}\,DT(3,0)-\frac{1}{8}\,DT(4,-1)+\frac13\, DT(4,-2)=0.
\end{split}
\ee
Barring a possible error in the table of GV invariants in \cite{Huang:2006hq}, we conclude that
the Ansatz \eqref{hpolar} does not produce the correct polar terms in this case. Given that it works in
many other cases, one might try to modify it in a minimal fashion, by changing just one or two
polar coefficients so as to restore modularity.
For example, it turns out that if one replaces $\hpol_{1,4}$ in \eqref{polarX42} by
\be
\hpol_{1,4} = \q^{-\frac{8}{3}} \left( {0_{-3}+(32+k)_{-2}\q-(2152+2k)_{-1}\q^2}\right),
\qquad
k\in\IZ,
\ee
one does find a modular form with integer coefficients.
Two choices of $k$ seem to be particularly interesting. If $k=0$, only one polar coefficient is changed, while
if $k=-20$, one ends up with the last polar coefficient given by $-2112=DT(4,-1)$,
which differs by the coefficient 1/3 from the Ansatz \eqref{hpolar}. However, besides these numerical observations,
we do not have any physical arguments in favor of one of these choices, and it may well be that
more than one polar coefficient is incorrectly predicted by our Ansatz in this case.

\subsection*{$\mathbf{X_{3,2,2}}$}

The lowest DT invariants are as follows:
\be
{\scriptsize
\begin{array}{|c|cccccccccc} \hline
Q \backslash n & -3 & -2 & -1 & 0 & 1 & 2 & 3 \\ \hline
0 &  0 & 0 & 0 & 1 & 144 & 10008 & 446304 \\
1 &0 & 0 & 0 & 0 & 720 & 102240 & 7000560 \\
2&  0 & 0 & 0 & 0 & 22428 & 3443616 & 254303604 \\
3&  0 & 0 & 0 & 64 & 1620720 & 245622240 & 18019908288 \\
4& 0 & 0 & 0 & 265113 & 206421552 & 27955859922 & 1957624164576 \\
5&  0 & 0 & 10080 & 199558944 & 50497608240 & 5249855378592 & 323810241865488 \\
6& -56 & -12096 & 179713440 & 115538513824 & 18048558130992 & 1472617884239424 &
   78052676370951268 \\ \hline
\end{array}
} \nn
\ee
Our Ansatz \eqref{hpolar} implies the following polar terms
\be
\begin{split}
\hpol_{1,0} =&\, \q^{-3} \left( {7_0 - 864_1 \q + 50040_2 \q^2}\right),
\\
\hpol_{1,1} =&\,\q^{-3+\frac{13}{24}} \left( {0_0 + 3600_1 \q - 408960_2 \q^2}\right),
\\
\hpol_{1,2} =&\,\q^{-3+\frac16} \left({0_{-1} + 0_0 \q -89712_1 \q^2}\right),
\\
\hpol_{1,3} =&\,\q^{-3+\frac78} \left({0_{-1} -256_0 \q + 4862160_1 \q^2}\right),
\\
\hpol_{1,4} =&\, \q^{-3+\frac23} \left( {0_{-2} +0_{-1} \q +795339_0 \q^2}\right),
\\
\hpol_{1,5} =&\,\q^{-3+\frac{13}{24}} \left( {0_{-3}+0_{-2} \q+30240_{-1} \q^2}\right),
\\
\hpol_{1,6} =&\,\q^{-3+\frac{1}{2}} \left( {0_{-4}+224_{-3}\q -36288_{-2} \q^2}\right).
\end{split}
\label{polarX322}
\ee
However, they fail to satisfy the constraint imposed by modularity, which would
require that the DT invariants fulfill the relation
\be
\begin{split}
& DT(0,0) -\frac{8}{21}\,DT(0, 1) - \frac{5}{21}\,DT(0, 2) +\frac{25}{126}\,DT(1, 1) +\frac{1}{63}\,DT(1, 2)
-\frac{16}{63}\, DT(2, 1)
\\
&  + \frac{1}{7}\, DT(3, 0) +\frac{1}{84}\,DT(3, 1) -\frac{1}{21}\, DT(4, 0) +\frac{1}{21}\,DT(6, -2)-\frac{8}{84}\,DT(5, -1).
\end{split}
\ee
As in the previous case, one might try to change just one or two polar coefficients so as to restore modularity.
For example, it turns out that if one replaces $\hpol_{1,6}$ in \eqref{polarX322} by
\be
\hpol_{1,6} = \q^{-3+\frac{1}{2}} \left( {0_{-4}+(224+k)_{-3}\q-12096_{-2}\q^2}\right),
\qquad
k\in\IZ,
\ee
one does find a modular form with integer coefficients.
Note that for such modification the last polar coefficient is given by $-12096=DT(6,-2)$,
which, like in the case of $X_{4,2}$, differs from the Ansatz \eqref{hpolar} by a coefficient 1/3.
However, this could just be a coincidence.

\subsection*{$\mathbf{X_{2,2,2,2}}$}

The lowest DT invariants are as follows:
\be\hspace{-0.5cm}
{\scriptsize
\begin{array}{|c|cccccccc} \hline
 Q \backslash n & -4 & -3 & -2 & -1 & 0 & 1 & 2 & 3 \\ \hline
 0 &  0 & 0 & 0 & 0 & 1 & 128 & 7872 & 308992 \\
 1 & 0 & 0 & 0 & 0 & 0 & 512 & 64512 & 3900928 \\
 2 & 0 & 0 & 0 & 0 & 0 & 9728 & 1356544 & 90337792 \\
 3 & 0 & 0 & 0 & 0 & 0 & 416256 & 57428992 & 3811304448 \\
 4 & 0 & 0 & 0 & 0 & 14752 & 27592192 & 3615258880 & 233963061760 \\
 5 & 0 & 0 & 0 & 0 & 8782848 & 3089741312 & 334005965824 & 19901940605440 \\
 6&0 & 0 & 0 & 1427968 & 2857640448 & 528800790528 & 44911222707968 & 2345453425978368 \\
7 &  0 & 0 & 86016 & 2451858432 & 934638858240 & 116559621707264 & 8004013269150720 &
   363671494077060608 \\
8 & -672 & -129024 & 2392944768 & 1945381563648 & 356833378589872 & 32067803814853376 &
   1801967963699774848 & 71093859294029974016 \\ \hline
\end{array}
 } \nn
\ee
Our Ansatz \eqref{hpolar} implies the following polar terms
\be
\begin{split}
\hpol_{1,0} =&\, \q^{-\frac{10}{3}} \left({-8_0 + 896_1 \q - 47232_2 \q^2 + 1544960_3 \q^3}\right) ,
\\
\hpol_{1,1} =&\, \q^{-\frac{10}{3}+\frac{17}{32}} \left({0_0 - 3072_1 \q + 322560_2 \q^2}\right) ,
\\
\hpol_{1,2} =&\, \q^{-\frac{10}{3}+\frac18} \left( {0_{-1} + 0_0 \q +48640_1 \q^2 - 5426176_2 \q^3} \right) ,
\\
\hpol_{1,3} =&\, \q^{-\frac{10}{3}+\frac{25}{32}} \left( {0_{-1}+0_0 \q-1665024_1 \q^2}\right) ,
\\
\hpol_{1,4} =&\, \q^{-\frac{10}{3}+\frac12} \left( {0_{-2}+0_{-1}\q-59008_0 \q^2}\right) ,
\\
\hpol_{1,5} =&\, \q^{-\frac{10}{3}+\frac{9}{32}} \left( {0_{-3}+0_{-2}\q+0_{-1}\q^2+26348544_0 \q^3}\right) ,
\\
\hpol_{1,6} =&\, \q^{-\frac{10}{3}+\frac18} \left( {0_{-4}+0_{-3}\q+0_{-2}\q^2+4283904_{-1} \q^3}\right) ,
\\
\hpol_{1,7} =&\, \q^{-\frac{10}{3}+\frac{1}{32}} \left( {0_{-5}+0_{-4}\q+0_{-3}\q^2+258048_{-2} \q^3}\right)  ,
\\
\hpol_{1,8} =&\, \q^{-\frac{10}{3}} \left( {0_{-6}+0_{-5}\q+2688_{-4}\q^2-387072_{-3}\q^3}\right).
\end{split}
\label{polarX2222}
\ee
However, they fail to satisfy three constraints imposed by modularity in this case.
One can again find a modular form with integer coefficients by appropriately modifying the polar terms.
In contrast to the cases of $X_{4,2}$ and $X_{3,2,2}$, it appears that we have to modify
at least 4 coefficients appearing in $\hpol_{1,\mu}$ with $\mu=6,7,8$ for such a solution to exist.
 For example,  if one replaces these functions in \eqref{polarX2222} by
\be
\begin{split}
\hpol_{1,6} =&\,\q^{-\frac{10}{3}+\frac18} \left( {0_{-4}+0_{-3}\q+0_{-2}\q^2+(2674832+440 k)_{-1} \q^3}\right) ,
\\
\hpol_{1,7} =&\, \q^{-\frac{10}{3}+\frac{1}{32}} \left( {0_{-5}+0_{-4}\q+0_{-3}\q^2-(469056-32 k)_{-2} \q^3}\right)  ,
\\
\hpol_{1,8} =&\, \q^{-\frac{10}{3}} \left( {0_{-6}+0_{-5}\q+(2690+23 k)_{-4}\q^2-(366336+128 k)_{-3}\q^3}\right),
\qquad
k\in\IZ,
\end{split}
\ee
one does find a modular form with integer coefficients. Of course, modularity could also
be restored by an even more drastic modification of the polar coefficients.

\section{Comparison with mathematical results \label{sec_math}}

In  \cite{Feyzbakhsh:2022ydn}, explicit formulae for rank 0 DT invariants are proven for
any smooth polarized CY threefold $\CY$ satisfying a technical condition known as
the Bogomolov-Gieseker inequality \cite{bayer2011bridgeland},
which is known to hold for the quintic $X_5$ \cite{li2019stability}, $X_6$, $X_8$ \cite{koseki2020stability}
and for $X_{4,2}$ \cite{liu2021stability}. A somewhat less explicit formula was proven
earlier for one-parameter CY threefolds in \cite[Thm 3.18]{Toda:2011aa}.
In this section, we translate Thm 1.1 of \cite{Feyzbakhsh:2022ydn} in our notations,
and compare to our Ansatz \eqref{hpolargenr} for the polar terms. We refrain from discussing
Thm 1.2 in {\it loc. cit.}, as it it is different in spirit from our Ansatz, but we anticipate that
it may also give valuable information on D4-D2-D0 indices \cite{followup}.

Let $v\in K(X)$ be a  rank-zero dimension-two class with Chern character
\be
(\ch_0,\ch_1,\ch_2,\ch_3)(v) =(0,D,\beta,m)\equiv v
\ee
with $D\neq 0$. Let
\be
Q_H(v) = \frac12 \left( \frac{ D\cdot H^2}{H^3 }\right)^2 +
6 \left(\frac{\beta\cdot H}{D\cdot H^2}\right)^2 -\frac{12 m}{D\cdot H^2}\, ,
\ee
where $H=c_1(\cO_\CY(1))$ is the polarization. According to \cite[Thm 1.1]{Feyzbakhsh:2022ydn},
$H$-Gieseker semistable sheaves of class $v$ can only exist only if $Q_H(v)\geq 0$.
Moreover, when $v$ satisfies
\be
\label{condFeyz}
(H^3)^2\, Q_H(v) < D\cdot H^2 -\frac52 + \frac{2}{D\cdot H^2}  - \frac{2}{(D\cdot H^2)^2}
\ee
then the DT invariant counting   $H$-Gieseker semistable sheaves of class $v$
is given by the explicit formula\footnote{In transcribing \cite[Thm 1.1]{Feyzbakhsh:2022ydn}, we
exchanged $v_1$ and $v_2$, set (after the exchange) $n_i=(-1)^{i-1}m_i$, and
denoted $DT(\beta\cdot \omega_a,n)={\mathrm I}_{n,\beta}$
and  $PT(\beta\cdot \omega_a,n)={\mathrm P}_{n,\beta}$ where $\omega_a$ is a basis in $H^2(\CY,\IZ)$.}
\be
\label{eqFeyz}
J(v)= \left(\sharp H^2(\CY,\IZ)_{\rm tors}\right)^2
\sum_{v_1+v_2=v\atop v_i\in M_i(v)}
(-1)^{\chi(v_1,v_2)-1} \chi(v_1,v_2)\, DT(\beta_1\cdot \omega_a,n_1)\,PT(\beta_2\cdot \omega_a,n_2),
\ee
where
\be
\begin{split}
v_1=&\, e^{D_1}(1,0,-\beta_1,-n_1),
\\
v_2=&\, -e^{D_2}(1,0,-\beta_2,n_2),
\end{split}
\label{defvv}
\ee
$\chi(-,-)$ is the Euler form given by
\bea
\chi(E,E') &=& \int_{\CY} \ch(E^*) \ch(E')\, \Td\CY
\\
&=& \ch_0 \left( \ch'_3 +\frac{1}{12}\,c_2(T\CY) \ch'_1\right)
- \ch'_0  \left(  \ch_3  +\frac{1}{12}\, c_2(T\CY) \ch_1\right) - \ch_1 \ch'_2 +  \ch_2 \ch'_1,
\nn
\eea
while $PT(Q_a,n)$ are the Pandharipande-Thomas (PT) invariants,
given by the same generating series as \eqref{gvc2} without the Mac-Mahon factor,
\be
Z_{PT}(\xi^a,q) =  \sum_{Q_a,n} PT(Q_a,n)\,e^{2 \pi \I Q_a \xi^a}  q^{n}
=  [M(-q)]^{- \chi_{\scriptstyle\CY}} Z_{DT} (\xi^a,q).
\ee
The sum in \eqref{eqFeyz} runs over  $(D_i,\beta_i,n_i)$, $i=1,2$ satisfying the inequalities
\be
\label{ineqFeyz}
\begin{split}
&
\frac12 \left( \frac{ D_i\cdot H^2}{H^3 }\right)^2 - \frac{D_i\cdot D_i\cdot H}{2H^3} +
\frac{\beta_i\cdot H}{H^3}\leq \frac{D \cdot H^2-2}{2(H^3)^2}\, ,
\\
& \qquad\qquad
n_i \geq -\frac{D\cdot H^2(D\cdot H^2+H^3)}{6 (H^3)^2}\, .
\end{split}
\ee
By the Grothendieck-Lefschetz theorem (see e.g. \cite[Ch. IV]{hartshorne2006ample}),
 $H^2(\CY,\IZ)$ is torsion free for any complete intersection
in a smooth projective variety, so the prefactor
$\left(\sharp H^2(\CY,\IZ)_{\rm tors}\right)^2$ in \eqref{eqFeyz} is trivial for the models
$X_5,X_{3,3},X_{4,2},X_{3,2,2},X_{2,2,2,2}$ considered in this paper. For the other models, the
ambient weighted projective space is singular, and $H^2(\CY,\IZ)$ could have non-trivial torsion
\cite{ravindra2005grothendieck}. We leave the determination of this factor as an open problem.

In  the notations of \S\ref{subsec-polar1},
one has
\be
\begin{split}
D =&\, r H,
\qquad\qquad
\qquad \beta\cdot H= -\mu-\frac{\kappa r^2}{2}\,,
\qquad
m
=\frac{\kappa r^3}{6} -n ,
\\
D_i =&\, (-1)^{i-1} r_i H,
\qquad\
\beta_i\cdot H= m_i,
\end{split}
\ee
with $H^3=\kappa$, such that
\be
Q_H(v) = \frac{12}{\kappa r} \left( \frac{\chi(r \cD)}{24} - \hat q_0 \right),
\ee
\be
\chi(v_1,v_2)=\cI_r-r(m_1+m_2)-n_1-n_2\equiv - \gamma_{12}.
\label{chigam}
\ee
The bound $Q_H(v)\geq 0$ is then recognized as the Bogomolov bound \eqref{qmax}, while the condition
\eqref{condFeyz} for the validity of \eqref{eqFeyz} becomes
\be
\label{ineqFeyz2}
\frac{12 \kappa}{r}  \left( \frac{\chi(r \cD)}{24} - \hat q_0 \right)
< \kappa r -\frac52 + \frac{2}{\kappa r} - \frac{2}{(\kappa r)^2}\, .
\ee
These two conditions may be written more compactly as
\be
\label{ineqFeyz3}
\frac{\kappa r^3}{24}\,  A(\kappa r) < \hat q_0 - \frac{\chi(r \cD)}{24} \leq 0\, ,
\ee
where $A(x) = (1-\frac{1}{x}+\frac{2}{x^2})^2-1$. The function $A(x)$ is positive
for $x<2$ and negative for $x>2$. It has a minimum $A=-\frac{15}{64}$ at $x=4$ and asymptotes to $-2/x$ as $x$ becomes large.
Thus, the range of validity of the formula \eqref{eqFeyz} shrinks as $\kappa$ and $r$ increase,
and is empty for $\kappa=r=1$.

To rewrite the formula \eqref{eqFeyz} in our notations, first let us write explicitly the condition $v=v_1+v_2$.
It leads to the following three equalities
\be
\label{v12}
\begin{split}
r_1+r_2=&\, r,
\\
m_1-m_2=&\, \mu+\kappa rr_1,
\\
n_1+n_2=&\, n-r_1 m_1-r_2 m_2-\frac{\kappa}{2}\, rr_1r_2.
\end{split}
\ee
The first two equations admit an integer solution for $r_1,r_2$ provided
\be
\label{m12mod}
m_1-m_2=\mu \mod \kappa r
\ee
After eliminating $r_1,r_2$,
the last relation in \eqref{v12} requires
\be
n_1+n_2=n+\frac{1}{2\kappa r}\,\Bigl(\mu^2-(m_1-m_2)^2+\kappa r^2(\mu-m_1-m_2)\Bigr).
\label{sumni}
\ee
Finally,  the inequalities \eqref{ineqFeyz} take the simple form
\be
m_i \leq  \frac{r}2 -\frac{1}{\kappa}\,,
\qquad
n_i  \geq -\frac16\, r(r+1)\, .
\label{ineqDTPT}
\ee
As a result, the formula \eqref{eqFeyz} takes the following form
\be
\label{eqFeyz2}
\bOm_{r,\mu}(\hq_0)=
\sum_{m_i, n_i}
(-1)^{\gamma_{12}}\gamma_{12}
\, DT(m_1,n_1)\, PT(m_2,n_2),
\ee
where $\gamma_{12}$ is given in \eqref{chigam}, and the sum is subject to the conditions  \eqref{m12mod}, \eqref{sumni}, \eqref{ineqDTPT} 
and $m_i\geq 0$. Since the condition \eqref{sumni} bounds $n_1,n_2$ from above,
the sum is manifestly finite.

The formula \eqref{eqFeyz2} is reminiscent of our Ansatz \eqref{hpolar} for $r=1$, or \eqref{hpolargenr} for higher $r$.
Note however that for the models considered in this paper,
the condition \eqref{ineqFeyz3} is never satisfied except for the most polar term
where $\hat q_0=\frac{\chi(r \cD)}{24}$ (assuming $\kappa\geq 2$), in which case $\mu=m_i= n_i=0$.


\providecommand{\href}[2]{#2}\begingroup\raggedright\endgroup

\end{document}